\documentclass{article}
\usepackage{graphicx,amssymb,amsmath,hyperref}

\newtheorem{theorem}{Theorem}[section]
\newtheorem{proposition}[theorem]{Proposition}
\newtheorem{definition}[theorem]{Definition}
\newtheorem{corollary}[theorem]{Corollary}
\newtheorem{lemma}[theorem]{Lemma}
\newtheorem{remark}[theorem]{Remark}
\newtheorem{example}[theorem]{Example}
\newenvironment{proof}{\noindent \emph{Proof. }}{\hfill \hbox{\rlap{$\sqcap$}$\sqcup$}\\}

\title{Brun Expansions of Stepped Surfaces}
\author{Val\'erie Berth\'e\footnote{Univ. Paris 7, CNRS, Sorbonne Paris Cit\'e, UMR 8243, 75205 Paris, France} \and Thomas Fernique\footnote{Univ. Paris 13, CNRS, Sorbonne Paris Cit\'e, UMR 7030, 93430 Villetaneuse, France.}}
\date{}

\begin{document}

\maketitle

\begin{abstract}
\emph{Dual maps} have been introduced as a generalization to higher dimensions of word substitutions and free group morphisms.
In this paper, we study the action of these dual maps on particular discrete planes and surfaces -- namely \emph{stepped planes} and \emph{stepped surfaces}.
We show that dual maps can be seen as discretizations of toral automorphisms.
We then provide a connection between stepped planes and the \emph{Brun} multi-dimensional continued fraction algorithm, based on a desubstitution process defined on local geometric configurations of stepped planes.
By extending this connection to stepped surfaces, we obtain an effective characterization of stepped planes (more exactly, \emph{stepped quasi-planes}) among stepped surfaces.
\end{abstract}

\medskip
{\small
\parbox{0.92\textwidth}{
{\bf Keywords:} arithmetic discrete plane, Brun algorithm, digital planarity, discrete geometry, dual map, flip, free group morphism, multidimensional continued fraction, stepped plane, stepped surface, substitution.}}

\section{Introduction}

In word combinatorics, Sturmian words and regular continued fractions are known to provide a very fruitful interaction between arithmetics, discrete geometry and symbolic dynamics.
Recall that Sturmian words can be defined as infinite words which code irrational discrete lines over a two-letter alphabet (one speaks about digitizations of irrational straight lines).
Recall also that a substitution is a non-erasing morphism of the free monoid which acts naturally on all finite and infinite words.
Then, most combinatorial properties of Sturmian words can be described in terms of the continued fraction expansion of the slope of the discrete line that they code (see Chap. 2 in \cite{Loth} and Chap. 6 in \cite{Pyt}).
For example, let us briefly sketch the proof that Sturmian words can be obtained as an infinite composition of a finite number of substitutions, \emph{i.e.}, Sturmian words are $S$-adic (for more details, see \cite{durand} and Chap. 12 in \cite{Pyt}).
First, one can deduce from the combinatorial properties of Sturmian words defined over $\{0,1\}$ that factors $00$ and $11$ cannot occur simultaneoulsy in a Sturmian word.
This allows to \emph{desubstitute} any Sturmian word $u$, \emph{i.e.} to write $u=\sigma_0(v)$ or $u=\sigma_1(v)$, where $v$ is an infinite word over $\{0,1\}$, and $\sigma_0$ and $\sigma_1$ are the substitutions defined by $\sigma_0(0)=0$, $\sigma_0(1)=10$, $\sigma_1(0)=01$ and $\sigma_1(1)=1$.
Then, one can show that the desubstituted word $v$ is itself a Sturmian word (it corresponds to a digitization of the same line after a change of lattice basis).
We can thus reiterate the process \emph{ad infinitum}, and the corresponding sequence of substitutions $\sigma_0$ and $\sigma_1$ turns out to be determined by the continued fraction expansion of the slope of the initial Sturmian word.\\

In this paper, we would like to extend this interaction to higher dimensions.
We thus need to generalize the notions of free group morphisms (among them substitutions), of Sturmian words and to work with a multi-dimensional version of the Euclidean algorithm, that is, a multi-dimensional continued fraction algorithm.
Here, we show that this can be respectively done by \emph{dual maps}, \emph{stepped planes} and the \emph{Brun algorithm}.
This allows to provide a first step towards a multi-dimensional extension of the above interaction.\\

Dual maps have been introduced by Arnoux and Ito in \cite{AI} as a generalization to higher dimensions of free monoid and free group morphisms.
They are inspired by the geometrical formalism of \cite{IO93}, whose aim was to provide explicit Markov partitions for hyperbolic automorphisms of the torus associated with particular morphisms of the free group.
Indeed, iterations of dual maps generate stepped planes approximating the stable and unstable spaces of toral automorphisms.
They have already proved their efficiency for the construction of explicit Markov partitions \cite{AFHI}, for Diophantine approximation \cite{ifhy}, in the spectral study of Pisot substitutive dynamical systems \cite{BK,Pyt} or else in discrete geometry \cite{ABFJ}.\\

Stepped planes have been introduced in \cite{BV} as multi-dimensional Sturmian words: they are digitizations of real hyperplanes (see Remark~\ref{rem:digitization_linear_group1} below for more details).
Stepped surfaces have then been introduced in \cite{jamet} as a generalization of two-letter words: a stepped surface is defined as a union of facets of integer translates of the unit hypercube which is homeormophic to the antidiagonal plane (the hyperplane with normal vector $(1,\ldots,1)$). Hence, stepped planes are particular stepped surfaces, as Sturmian words are particular two-letter words.
Following \cite{ABFJ}, we also rely in this paper on the notion of \emph{flip}.
The flip is a classical notion in the study of dimer tilings: this is a local reorganization of tiles that transforms a tiling into another one (see, \emph{e.g.}, \cite{thurston}).
In our context, flips turn out to be a powerful technical tool that allows us to transfer properties from stepped planes to stepped surfaces, by describing the latter as stepped planes on which flips are performed.\\

Last, the Brun algorithm (also called modified Jacobi-Perron algorithm), introduced in \cite{brun}, is one of the most classical unimodular multi-dimensional continued fraction algorithms (in the sense of \cite{brentjes}, see also \cite{schweiger}).
Although we choose the Brun algorithm, many other algorithms as, \emph{e.g.}, the Jacobi-Perron, Selmer or Poincar\'e ones could be also used.\\

Let us outline the contents and the main results of the present paper.
First, Sec.\ \ref{sec:stepped} is devoted to the basic introductory material, namely stepped functions, which include both stepped planes and stepped surfaces.
Then, Sec.\ \ref{sec:flips} introduces the notion of flip in the general context of stepped functions.
This leads to the notion of \emph{pseudo-flip-accessibility}, which extends the usual notion of \emph{flip-accessibility}.\\
The next section, Sec.\ \ref{sec:dual_maps}, recalls the notion of dual map and provides two of the main results of this paper, whose proofs are combinatorial: the image of stepped planes and stepped surfaces under dual maps are, respectively, stepped planes and stepped surfaces (Th.\ \ref{th:image_stepped_plane} and Th.\ \ref{th:image_stepped_surface}).
This extends similar results that we have previously obtained in the particular case of \emph{positive} dual maps (which includes substitutions but not any free group morphism).
Last, Sec.\ \ref{sec:brun_expansions} relies on the results of the previous section to define Brun expansions of both stepped planes and stepped surfaces.
We first handle the case of a stepped plane: its Brun expansion indeed naturally corresponds to the Brun expansion of its normal vector.
However, we also provide a definition of the Brun expansion of a stepped plane which relies not on its normal vector but only on some of its local geometric configurations, namely \emph{runs} (see Def.\ \ref{def:run} and \ref{def:tilde_T}).
This allows us to then extend the notion of Brun expansion to stepped surfaces (even with a lack of a notion of a normal vector) (Th.\ \ref{th:tilde_T_stepped_surface} and Def.\ \ref{def:stepped_surface_exp}). 
We also here prove the following ``classification'' result: the longer the Brun expansion of a stepped surface is, the more planar this stepped surface is (Th.\ \ref{th:weak_cv_surfaces}).
Sec.\ \ref{sec:con} ends the paper with additional remarks.

\section*{Notation}

Let us provide here some notation used throughout this paper.
The dimension of the space is denoted by $d$, and we assume $d\geq 3$; we thus work in $\mathbb{R}^d$.
Let $\mathbb{K}$ be equal to $\mathbb{Z}$ or $\mathbb{R}$.
The set of non-zero (resp. non-negative) elements of $\mathbb{K}$ is denoted by $\mathbb{K}^*$ (resp. $\mathbb{K}_+$).
The set $\mathbb{N}$ denotes the set $\{0,1,2,\ldots\}$ of non-negative integers. We stress the fact that     the term  positive  refers  to  strictly positive in all that follows.
The cardinality of a set $X$ is denoted by $\#X$.
If $\vec{x}=(x_1,\ldots,x_d)$ and $\vec{y}=(y_1,\ldots,y_d)$ belong to $\mathbb{R}^d$, then we write $\vec{x} \leq \vec{y}$ (resp. $\vec{x} < \vec{y}$) if, for all $i\in \{1,\ldots,d\}$, $x_i \leq y_i$ (resp. $x_i <y_i$).
We also denote by $\langle\vec{x}|\vec{y}\rangle=\sum_ix_iy_i$ the scalar product of $\vec{x}$ and $\vec{y}$. The Euclidean norm in $\mathbb{R}^d$ is denoted by $||\,.\,||$, and $B(\vec{x},r)$ stands for the Euclidean closed ball of center $\vec{x}$ and radius $r$.
We also use the notation $||\,.\,||$ for the matrix norm associated with the Euclidean norm, \emph{i.e.}, $||M||=\sup_{\vec{x}\neq\vec{0}}(||M\vec{x}||/||\vec{x}||)$ for a $d\times d$ real matrix $M$.
Last, for $\vec{\alpha}\in\mathbb{R}^d$, $\vec{\alpha}^\bot$ stands for the hyperplane orthogonal to the real line $\mathbb{R}\vec{\alpha}$.

\section{Stepped functions}
\label{sec:stepped}

We here introduce basic objects of this paper, thanks to an algebraic formalism for unions of facets of integer translates of the unit cube.
We first introduce the notion of \emph{stepped function}:

\begin{definition}\label{def:stepped_function}
A \emph{stepped function} $\mathcal{E}$ is a function from $\mathbb{Z}^d\times\{1,\ldots,d\}$ to $\mathbb{Z}$.
Its \emph{size} is the cardinality of the subset of $\mathbb{Z}^d\times\{1,\ldots,d\}$ where it takes non-zero values.
The set of stepped functions is denoted by $\mathfrak{F}$.
\end{definition}

The set of stepped functions is a $\mathbb{Z}$-module.
For any two stepped functions $\mathcal{E}$ and $\mathcal{E}'$, we write $\mathcal{E}\leq\mathcal{E}'$ if $\mathcal{E}(\vec{x},i)\leq\mathcal{E}'(\vec{x},i)$ for any $(\vec{x},i)\in\mathbb{Z}^d\times\{1,\ldots,d\}$.
We then endow $\mathfrak{F}$ with the following metric:

\begin{definition}\label{def:distance}
Let $d_\mathfrak{F}$ be the distance on $\mathfrak{F}$ defined by $d_\mathfrak{F}(\mathcal{E},\mathcal{E}')=0$ if $\mathcal{E}=\mathcal{E}'$, and by $d_\mathfrak{F}(\mathcal{E},\mathcal{E}')=2^{-r}$ otherwise, where
$$
r=\max\{n\in\mathbb{N}~|~\forall (\vec{x},i)\in\mathbb{Z}^d\times\{1,\ldots,d\},~||\vec{x}||< n ~\Rightarrow~\mathcal{E}(\vec{x},i)=\mathcal{E}'(\vec{x},i)\}.
$$
\end{definition}

This metric is inspired by the rigid version of the so-called \emph{local metric} often used for tilings (see, \emph{e.g.}, \cite{RW,rob}). 
Note that the distance between two stepped functions which agree nowhere is equal to one ($r=0$).
The only non-trivial point to ensure that $d_\mathfrak{F}$ is a distance is that the triangle inequality holds.
One checks that, for any three stepped functions $\mathcal{E}$, $\mathcal{E}'$ and $\mathcal{E}''$, the following ultrametric inequality holds:
$$
d_\mathfrak{F}(\mathcal{E},\mathcal{E}')\leq \max (d_\mathfrak{F}(\mathcal{E},\mathcal{E}''),d_\mathfrak{F}(\mathcal{E}'',\mathcal{E}')).
$$

\noindent Among stepped functions, we distinguish the following elementary ones:

\begin{definition}\label{def:face}
The \emph{face of type $i\in\{1,\ldots,d\}$ located at $\vec{x}\in\mathbb{Z}^d$}, denoted by $(\vec{x},i^*)$, is the stepped function taking value one at $(\vec{x},i)$ and zero elsewhere.
\end{definition}

The notation $(\vec{x},i^*)$ allows one to distinguish between the function and the element $(\vec{x},i)\in\mathbb{Z}^d \times \{1,\ldots,d\} $.
Although it would be more natural to use the notation $(\vec{x},i)^* $, we prefer $(\vec{x},i^*)$ for simplicity.\\

\noindent Faces allow us to write any stepped function $\mathcal{E}\in\mathfrak{F}$ as an at most countable weighted sum of faces:
$$
\mathcal{E}=\sum_{(\vec{x},i)\in\mathbb{Z}^d\times\{1,\ldots,d\}}\mathcal{E}(\vec{x},i)(\vec{x},i^*),
$$
where $\mathcal{E}(\vec{x},i)$ is called the \emph{weight} of the face $(\vec{x},i^*)$.
This notation is convenient and will later be used (see, \emph{e.g.}, Def.\ \ref{def:stepped_quasi_plane}, below).
Let us stress the fact that such sums are formal and do not raise any problem of convergence.\\

Let us now provide a geometric interpretation of faces.
Let $(\vec{e}_1,\ldots,\vec{e}_d)$ denote the canonical basis of $\mathbb{R}^d$.
The {\em geometric interpretation} of the face $(\vec{x},i^*)$ is defined as the facet\footnote{Note that we choose here the facet containing $\vec{x}+\vec{e}_i$ and not $\vec{x}$ (in other words, we add an offset $\vec{e}_i$): this is only for compatibility with formulas of \cite{AI} used here in Sec.
\ref{sec:dual_maps}.} of unit hypercube of $\mathbb{R}^d$ (see Fig.\ \ref{fig:faces}):
$$
\{\vec{x}+\vec{e}_i+\sum_{j\neq i}\lambda_j\vec{e}_j~|~0\leq\lambda_j\leq 1\}.
$$

\begin{figure}[hbtp]
\centering
\includegraphics[width=0.6\textwidth]{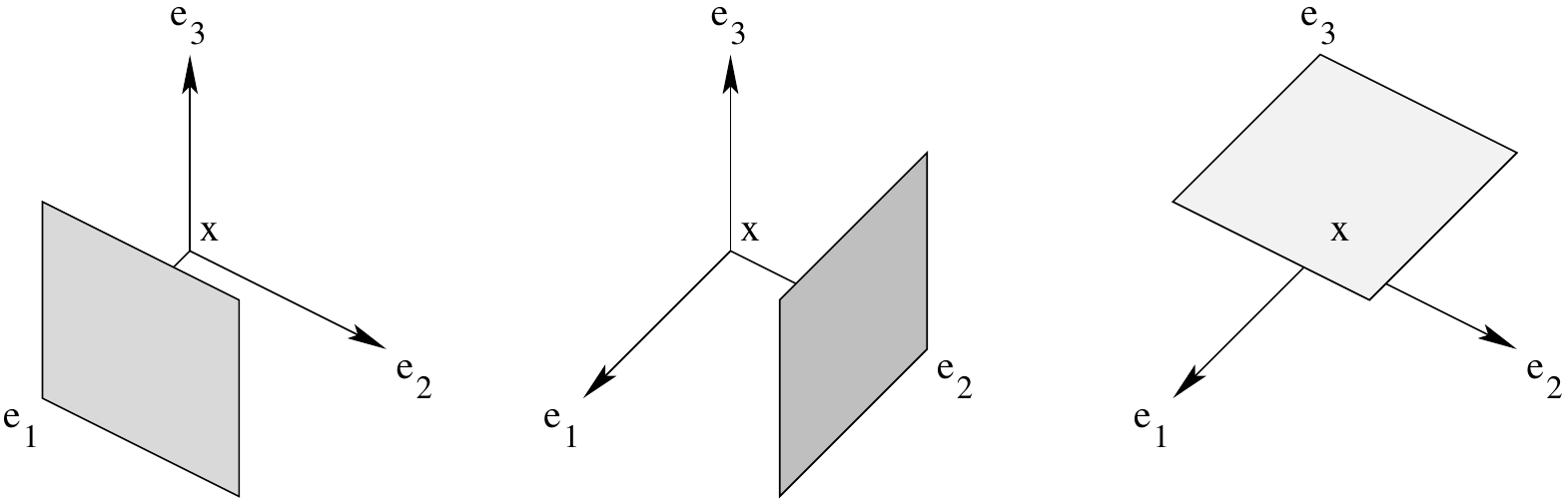}
\caption{Geometrical interpretations of faces $(\vec{x},i^*)$, for $i=1,2,3$ (from left to right).}
\label{fig:faces}
\end{figure}

Then, the geometric interpretation of a sum of faces whose weights are all equal to zero or one is naturally defined as the union of geometric interpretations of faces with weight one.
This leads to distinguish particular stepped functions:

\begin{definition}\label{def:binary_stepped_function}
A stepped function is said to be \emph{binary} if it takes only values zero or one.
The set of binary stepped functions is denoted by $\mathfrak{B}$.
\end{definition}

Thus, binary stepped functions are the stepped functions having a geometric interpretation (which is a union of facets of unit hypercubes).\\

We are now in a position to introduce two types of stepped functions playing a key role throughout this paper, namely \emph{stepped surfaces} and \emph{stepped planes}:

\begin{definition}\label{def:stepped_surface}
Let $\pi$ be the orthogonal projection onto the Euclidean hyperplane $\Delta=\{\vec{x}~|~x_1+\ldots+x_d=0\}$.
A binary stepped function whose geometric interpretation is homeomorphic to $\Delta$ under $\pi$ is called a \emph{stepped surface}.
The set of stepped surfaces is denoted by $\mathfrak{S}$.
\end{definition}

In other words, stepped surfaces correspond (by projecting under $\pi$ their geometrical interpretation) to tilings of $\Delta$ by $d$ types of rhomboedras ($3$ types of lozenges when $d=2$, see \emph{e.g.} Fig. \ref{fig:stepped_plane_surface} below).

\begin{definition}\label{def:stepped_plane}
The \emph{stepped plane} of \emph{normal vector} $\vec{\alpha}\in\mathbb{R}_+^d\backslash\{\vec{0}\}$ and \emph{intercept} $\rho\in\mathbb{R}$ is the binary stepped function denoted by $\mathcal{P}_{\vec{\alpha},\rho}$ and defined by:
$$
\mathcal{P}_{\vec{\alpha},\rho}(\vec{x},i)=1 ~\Leftrightarrow~ \langle\vec{x}|\vec{\alpha}\rangle < \rho\leq\langle\vec{x}+\vec{e}_i|\vec{\alpha}\rangle.
$$
The set of stepped planes is denoted by $\mathfrak{P}$.
\end{definition}

The set of vertices of the geometric interpretation of a stepped plane is usually called a standard arithmetic discrete plane in discrete geometry, according to the terminology of \cite{reveilles}.
Fig.\ \ref{fig:stepped_plane_surface} illustrates Def.\ \ref{def:stepped_surface} and \ref{def:stepped_plane}.

\begin{figure}[hbtp]
\centering
\includegraphics[width=\textwidth]{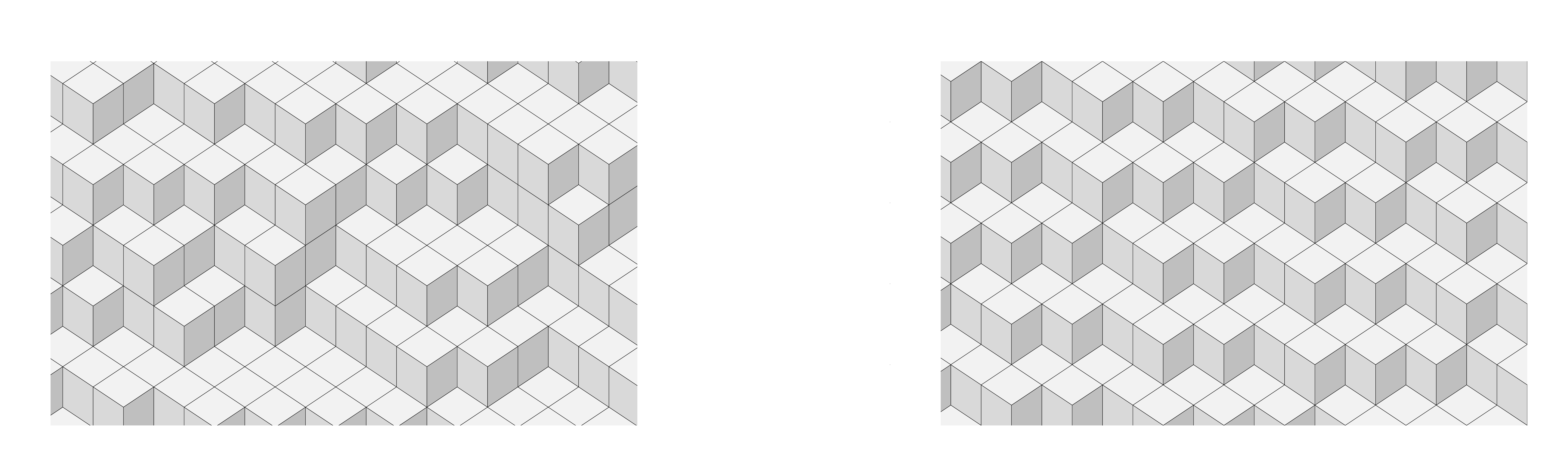}
\caption{Geometrical interpretation of a stepped surface (left) and of a stepped plane (right).
Both are unions of facets of unit hypercubes of $\mathbb{R}^d$, whose images under $\pi$ can be seen as tilings of the hyperplane $\Delta$ (here, $d=3$).
}
\label{fig:stepped_plane_surface}
\end{figure}

Note that the condition $\vec{\alpha}\in\mathbb{R}_+^d\backslash\{\vec{0}\}$ of Def.\ \ref{def:stepped_plane} ensures that the geometric interpretation of a stepped plane is homeomorphic to $\Delta$ under $\pi$.
Stepped planes are thus particular stepped surfaces, namely ``straight'' ones, and the different types of stepped functions previously introduced verify the following inclusions (with all of them being strict):
$$
\mathfrak{P}\subset\mathfrak{S}\subset\mathfrak{B}\subset\mathfrak{F}.
$$
One can also check that $\mathfrak{S}$, $\mathfrak{B}$ and $\mathfrak{F}$ are closed (w.r.t. the distance $d_\mathfrak{F}$), while $\mathfrak{P}$ is neither closed nor open.\\

We conclude this section by stating a technical property of stepped surfaces that will be used later (in the proofs of Prop.\ \ref{prop:positive_image_stepped_surface} and Lem.\ \ref{lem:zarbi}):

\begin{proposition}\label{prop:stepped_surface}
If $\vec{x}$ and $\vec{y}$ are two integer vectors belonging to the geometric interpretation of a stepped surface, then $\vec{x}-\vec{y}$ is neither positive nor negative.
\end{proposition}
\begin{proof}
Let $\mathcal{S}$ be a stepped surface, and 
 $S$ stand for its geometric interpretation.
We denote by $\vec{u}$ the vector $\vec{e}_1+\ldots+\vec{e}_d$.
Since ${S}$ is homeomorphic to the hyperplane $\Delta=\vec{u}^\bot$, $S$ divides $\mathbb{R}^d$ into two open halfspaces $S^+$ and $S^-$.
Let $S^+$ be the open halfspace in the direction $\vec{u}$.
One has:
$$
\forall \vec{z}\in S,~\forall k>0,~\vec{z}+k\vec{u}\in S^+,~\vec{z}-k\vec{u}\in S^-.
$$
Let us prove that if $\vec{x}\in S\cap\mathbb{Z}^d$, then $\vec{x}-\vec{e}_i \in S\cup S^-$, for $i\in \{1,\ldots,d\}$.
W.l.o.g., we can assume $\vec{x}=\vec{0}$.
Assume furthermore that $-\vec{e}_i \notin S$. We  want to prove that  $-\vec{e}_i \in S^-.$
There is $\vec{z} \in S$ which is mapped to $\pi(-\vec{e}_i)$ under $\pi$.
Hence there is $k \in \mathbb{Z}$ such that $\vec{z}= -\vec{e}_i + k\vec{u}$.
We want to prove that $k=1$. The claim will follow directly by noticing that
$-\vec{e}_i =( -\vec{e}_i + \vec{u} )-\vec{u}$ and $-\vec{e}_i + \vec{u}=\vec{z} \in S$, i.e.,  $-\vec{e}_i  \in S^-$.\\
Thus, let us prove that $k=1$. There is a facet $F$ included in $S$, with $F$ being the geometric interpretation of a face with weight one of $\mathcal{S}$, such that $\pi(F)\subset\Delta$ contains the line segment with ends $\pi(\vec{0})$ and $\pi(-\vec{e_i})$.
The set of vertices of $\pi(F)$ is included in $$\{\sum_{i=1}^d\delta_i\pi(\vec{e}_j)|\delta_i\in\{0,\pm 1\}\}.$$
 Indeed, for any $j\in \{1,\ldots,d\}$, the set of all vertices of $\pi(F)$ cannot contain both $\pi(\vec{e}_j)$ and $\pi(-\vec{e}_j)$. Hence, $\pi(-\vec{e}_i)$ can be reached from $\vec{0}$ by an edge path in $\pi(F)$.
If we lift up this path in $S$, we also obtain that $-\vec{e}_i +k \vec{u}$
can be reached from $\vec{0}$
by an edge path in $F$. (For an illustration of an example of such a path in the case of
$i=1$, see Fig.~\ref{fig:proof:stepped_surface}.)
Hence, $-\vec{e}_i +k \vec{u} \in \{\sum_{i=1}^d\delta_i \vec{e}_j|\delta_i\in\{0,\pm 1\}\},$
which implies that $k=1$.   This ends the proof of  the fact that  if $\vec{x}\in S\cap\mathbb{Z}^d$, then $\vec{x}-\vec{e}_i \in S\cup S^-$, for $i\in \{1,\ldots,d\}$.\\
Now, let $\vec{x} \in S^-\cap \mathbb{Z}^d$.
There is $\ell>0$ such that $\vec{x}+ \ell \vec{u} \in S$.
We deduce from above that $\vec{x}+ \ell \vec{u}-\vec{e}_i \in S^-$, and thus $\vec{x}-\vec{e}_i \in S \cup S^-$.
We thus have proved that if $\vec{x}\in (S\cup S^-) \cap \mathbb{Z}^d$, then $\vec{x}-\vec{e}_i \in S\cup S^-$, for $i\in \{1,\ldots,d\}$.
This inductively implies that if $\vec{x} \in S \cap \mathbb{Z}^d$, then, for any $\vec{v} \in \mathbb{N}^d$, $\vec{x}-\vec{v} \in S\cup S^-$.\\
Finally, let $\vec{x}$ and $\vec{y}$ in $S \cap \mathbb{Z}^d$ such that $\vec{x}-\vec{y}>0$.
One has $\vec{x} \geq \vec{y} + \vec{u}$, hence there is $\vec{v}\in \mathbb{N}^d$ such that $\vec{x}=\vec{y}+\vec{u}+\vec{v}$.
This yields $\vec{x}-\vec{v}= \vec{y} + \vec{u}\in S ^+$ and $\vec{x}-\vec{v} \in S\cup S^{-}$.
This contradicts $(S \cup S^-) \cap S^+=\emptyset$.
\end{proof}
\begin{figure}[hbtp]
\centering
\includegraphics[width=8cm]{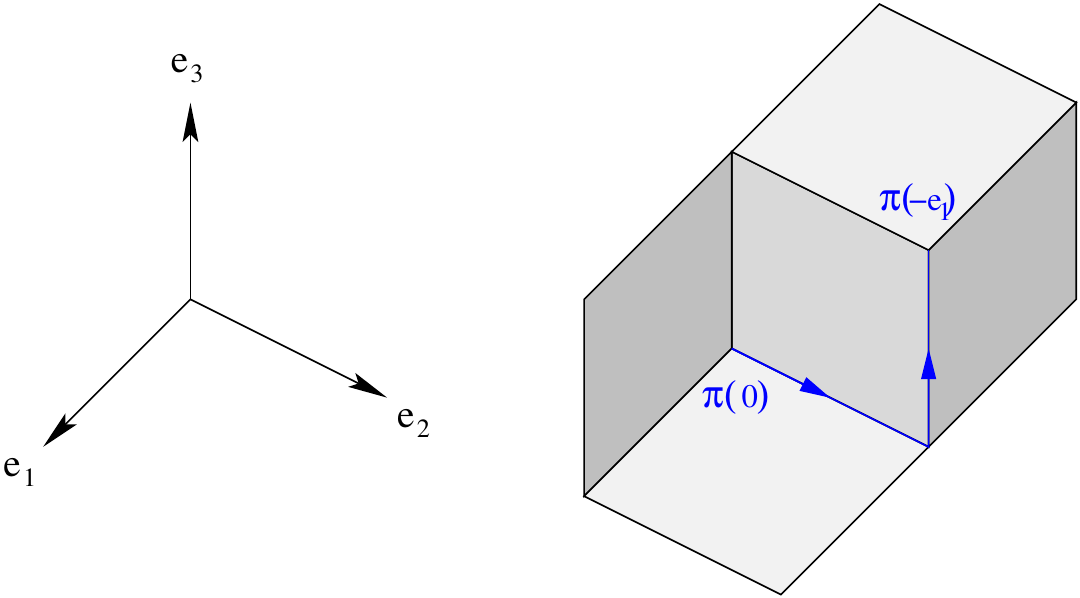}
\caption{A path in $S$ from $\vec{0}$ to $\pi(-\vec{e}_1)$ (here, $d=3$).
}
\label{fig:proof:stepped_surface}
\end{figure}

\section{Flips and pseudo-flips}
\label{sec:flips}

In mechanical physics, matter is often modeled by tilings, with local rearrangements of inter-atomic links being modeled by local rearrangements of tiles called \emph{flips}.
This rises questions about the structure of the space of tilings endowed with the flip operation, as well as about the dynamics of flips over this space.
In particular, the question of \emph{flip-accessibility} is natural: can one transform a tiling into another one by performing a sequence of flips?\\

In the previous section, we mentioned that the orthogonal projection $\pi$ onto $\Delta$ naturally associates with each stepped surface of $\mathbb{R}^d$ a tiling of $\Delta$ (see Fig.\ \ref{fig:stepped_plane_surface}).
Here, we take an $\mathbb{R}^d$-viewpoint and directly define a notion of flip on stepped surfaces, so that performing a flip on a stepped surface corresponds to performing a classic flip on the associated tiling.
More formally, we introduce the following stepped functions:

\begin{definition}\label{def:flip}
The \emph{flip} located at $\vec{x}\in\mathbb{Z}^d$ is the stepped function $\mathcal{F}_{\vec{x}}$ defined by:
$$
\mathcal{F}_{\vec{x}}:=\sum_{i=1}^{d}(\vec{x},i^*)-\sum_{i=1}^{d}(\vec{x}-\vec{e}_i,i^*).
$$
\end{definition}

Performing a flip on a stepped surface $\mathcal{S}$ thus means adding some $\mathcal{F}_{\vec{x}}$ to $\mathcal{S}$, provided $\mathcal{F}_{\vec{x}}+\mathcal{S}$ is still a stepped surface.
Then, transforming a tiling into another one by flips becomes pushing outwards the difference between the two corresponding stepped surfaces (see Fig.\ \ref{fig:flip}).
The precise definition of flip-accessibility in this setting is the following:

\begin{definition}\label{def:flip_acc}
A stepped surface $\mathcal{S}'$ is \emph{flip-accessible} from a stepped surface $\mathcal{S}$ if there is a sequence $(\mathcal{S}_n)_{n \in \mathbb{N}}$ of stepped surfaces such that:
$$
\mathcal{S}_0=\mathcal{S},\qquad
\mathcal{S}_{n+1}-\mathcal{S}_n\in\{\pm\mathcal{F}_{\vec{x}}~|~\vec{x}\in\mathbb{Z}^d\},\qquad
\lim_{n\to\infty}\mathcal{S}_n=\mathcal{S}'.
$$
\end{definition}

\begin{figure}[hbtp]
\centering
\includegraphics[width=0.7\textwidth]{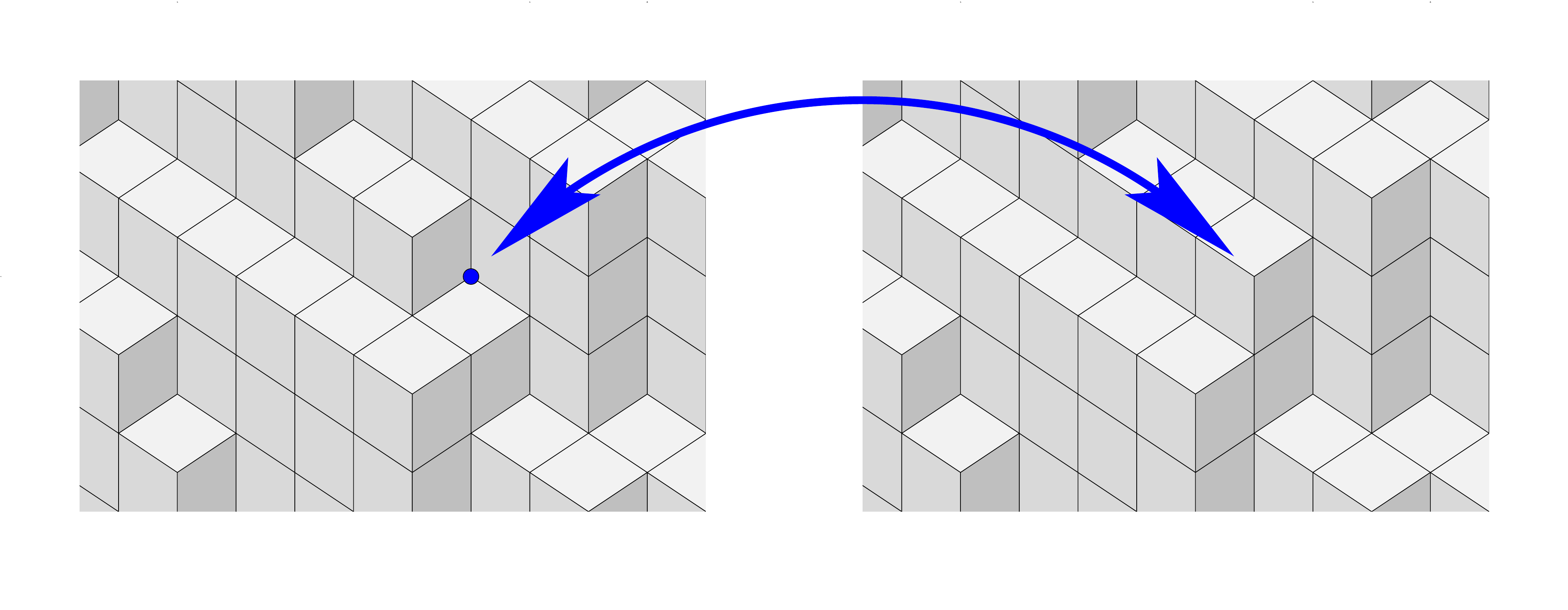}
\caption{Two stepped surfaces which differ by a flip (located at the dotted vertex).
Adding a flip can be seen in $\mathbb{R}^d$ as adding the upper facets of a unit hypercube to the geometric interpretation of a stepped surface or, once projected under $\pi$, as a local rearrangement of tiles in a tiling of $\Delta$.
}
\label{fig:flip}
\end{figure}

Note that the sequence of stepped surfaces $(\mathcal{S}_n)_{n \in \mathbb{N}}$ connecting two stepped surfaces can be finite as well as infinite.
It is easily seen that two stepped surfaces are not always flip-accessible.
For example, consider the stepped plane $\mathcal{P}_{\vec{e}_1,0}$: since it takes values one only on faces of type $1$, $\mathcal{P}_{\vec{e}_1,0}+\mathcal{F}_{\vec{x}}$ takes negative values, whatever $\vec{x}$ is.
Thus, no stepped surface is flip-accessible from $\mathcal{P}_{\vec{e}_1,0}$ (it is a sink).
Let us mention that a characterization of flip-accessible stepped surfaces in terms of \emph{shadows} is provided in \cite{ABFJ} (see also \cite{BFR}).
Furthermore, it is proven there that any stepped surface is flip-accessible from a stepped plane whose normal vector is positive (that is, each entry is  strictly positive).\\

Def.\ \ref{def:flip_acc} translates flip-accessibility from tilings to stepped surfaces.
But the formalism introduced in Sec.\ \ref{sec:stepped} allows us to define a more general notion of flip-accessibility, namely \emph{pseudo-flip-accessibility}, which involves not only stepped surfaces but also stepped functions:

\begin{definition}\label{def:pseudo_flip_acc}
A stepped function $\mathcal{E}'$ is \emph{pseudo-flip-accessible} from a stepped function $\mathcal{E}$ if there is a sequence $(\mathcal{E}_n)_{n \in \mathbb{N}}$ of stepped functions such that:
$$
\mathcal{E}_0=\mathcal{E},\qquad
\mathcal{E}_{n+1}-\mathcal{E}_n\in\{\pm\mathcal{F}_{\vec{x}}~|~\vec{x}\in\mathbb{Z}^d\},\qquad
\lim_{n\to\infty}\mathcal{E}_n=\mathcal{E}'.
$$
\end{definition}

While flip-accessibility connects two stepped surfaces by a sequence of stepped surfaces, such a restriction does no longer hold on the connecting sequence for pseudo-flip-accessibility (\emph{i.e.}, the $\mathcal{E}_n$ appearing in Def.\ \ref{def:pseudo_flip_acc} are not necessarily stepped surfaces, even if $\mathcal{E}$ and $\mathcal{E}'$ do).
In particular, characterizing pseudo-flip-accessibility between stepped surfaces becomes a trivial question:

\begin{proposition}\label{prop:pseudo_flip_acc}
Any two stepped surfaces are mutually pseudo-flip-accessible.
\end{proposition}

\begin{proof}
We use the notation of the proof of Prop.~\ref{prop:stepped_surface}.
We consider two stepped surfaces $\mathcal{S}$ and $\mathcal{S}'$.
Let $\varepsilon_{\mathcal{S},\mathcal{S}'}:\mathbb{Z}^d\to\{0,\pm 1\}$ be defined by:
$$
\varepsilon_{\mathcal{S},\mathcal{S}'}(\vec{x})=\left\{\begin{array}{rl}
1 & \textrm{if $\vec{x} \in S \cup S^+$ and $\vec{x} \in ({S}')^-$,}\\
-1 & \textrm{if $\vec{x} \in S^-$ and $\vec{x} \in {S'} \cup (S')^+$,}\\
0 & \textrm{otherwise.}
\end{array}\right.
$$
We label the elements of the countable set $\{\vec{x}\in\mathbb{Z}^d~|~\varepsilon_{\mathcal{S},\mathcal{S}'}(\vec{x})\neq 0\}$ in such a way that $||\vec{x}_i||\leq||\vec{x}_j||$ whenever $i<j$.
This ensures the convergence of the sequence $(\mathcal{E}_n)$ of stepped functions (which have no reason to be stepped surfaces) defined, for any $n$, by:
$$
\mathcal{E}_n=\mathcal{S}+\sum_{k\leq n} \varepsilon_{\mathcal{S},\mathcal{S}'}(\vec{x}_k)\mathcal{F}_{\vec{x}_k}.
$$
Moreover, this sequence tends towards $\mathcal{S}'$.
Indeed, while a vector $\vec{x}\in\mathbb{Z}^d$ of $\mathcal{S}$ is still above (resp. below) $\mathcal{S}'$, it will be moved downwards (resp. upwards) by a flip (according to the sign of $\varepsilon_{\mathcal{S},\mathcal{S}'}(\vec{x})$).
This shows that $\mathcal{S}'$ is pseudo-flip-accessible from $\mathcal{S}$.
\end{proof}

\noindent Fig.\ \ref{fig:pseudo_flip_acc} illustrates the proof of Prop.\ \ref{prop:pseudo_flip_acc}, and more precisely how the labeling and flipping is performed.\\

\begin{figure}[hbtp]
\centering
\includegraphics[width=0.6\textwidth]{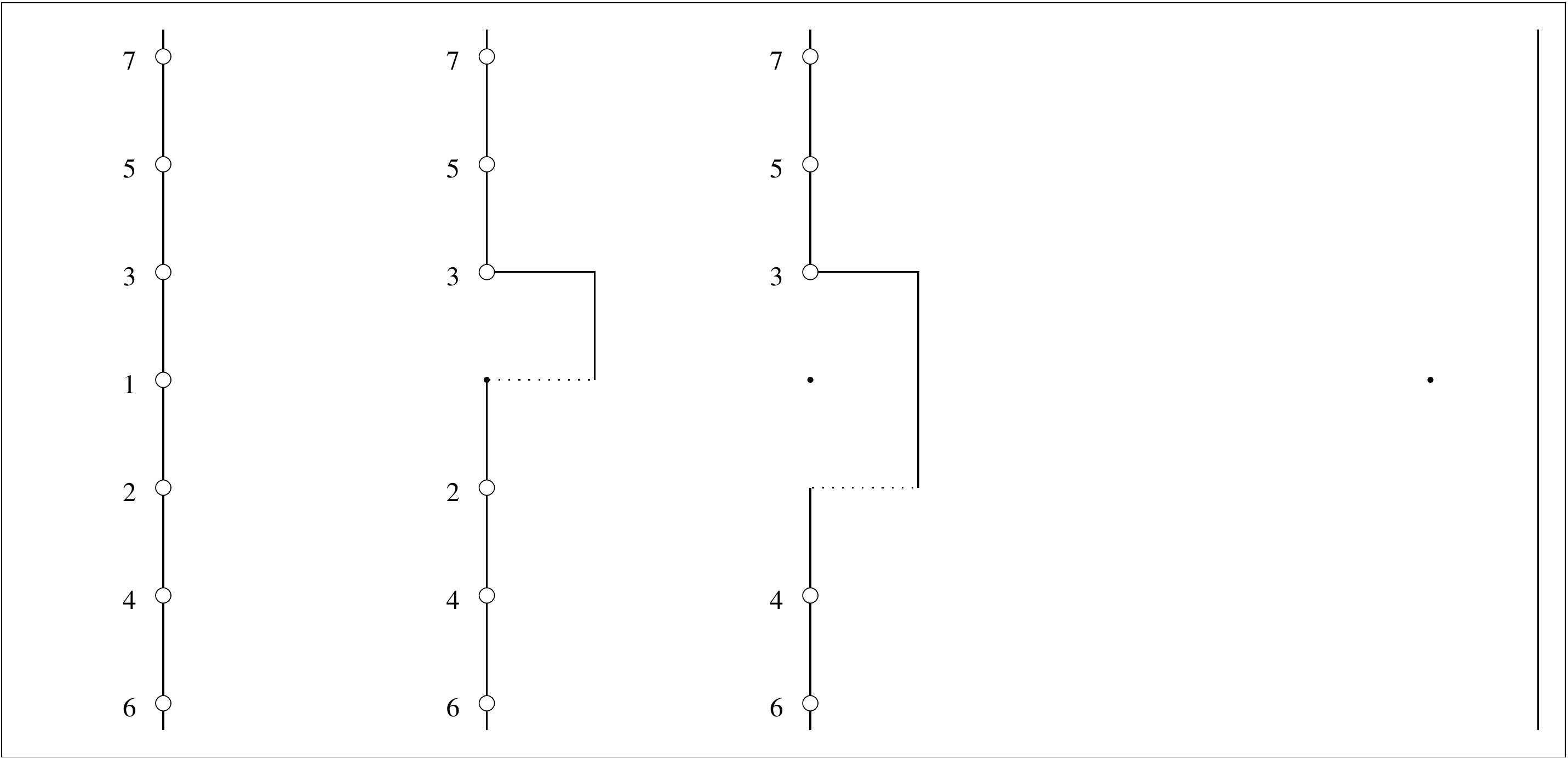}
\caption{
Assume that we want to transform, by performing flips, the stepped surface $\mathcal{S}=\sum_{k\in\mathbb{Z}}(k\vec{e}_2,1^*)$ into the stepped surface $\mathcal{S}'=\sum_{k\in\mathbb{Z}}(\vec{e}_1+k\vec{e}_2,1^*)$, \emph{i.e.} to translate it by $\vec{e}_1$.
Following the proof of Prop. \ref{prop:pseudo_flip_acc}, we add flips on white points (where $\varepsilon_{\mathcal{S},\mathcal{S}'}$ is equal to $1$).
These flips are sorted by increasing distance from the origin, as indicated on the figure.
By performing these flips, we thus transform $\mathcal{S}$ into $\mathcal{S}'$ (from left to right).
Note that intermediary stepped functions are not necessarily stepped surfaces (dashed edges stand for faces with weight $-1$).
}
\label{fig:pseudo_flip_acc}
\end{figure}

One can ask whether pseudo-flip-accessibility is too general for being worthwhile.
We will however use it in the proofs of Th.\ \ref{th:image_stepped_surface} and
 \ref{th:weak_cv_surfaces}, by decomposing a stepped surface into a sum of a stepped plane and flips.
We will in particular need the following result:

\begin{proposition}\label{prop:binary_pseudo_flip_acc}
If a non-zero binary stepped function is pseudo-flip-accessible from a stepped surface, then it is a stepped surface.
\end{proposition}

\begin{proof}
Let $\mathcal{B}$ be a non-zero binary function pseudo-flip-accessible from a stepped surface $\mathcal{S}$.
Thus, there is a sequence $(\vec{x}_n)_{n\geq 0}$ with values in $\mathbb{Z}^d$ and a sequence $(\lambda_n)_{n\geq 0}$ with values in $\mathbb{Z}^*$ such that the sequence of stepped functions $(\mathcal{E}_n)_{n\geq 0}$ defined as follows tends towards $\mathcal{B}$:
$$
\mathcal{E}_n:=\mathcal{S}+\sum_{0\leq k\leq n}\lambda_k\mathcal{F}_{\vec{x}_k}.
$$
W.l.o.g., the $\vec{x}_n$'s are pairwise different.
We then introduce a second sequence of stepped functions:
$$
\mathcal{E}'_n:=\mathcal{S}+\sum_{||\vec{x}_k||\leq n}\lambda_k\mathcal{F}_{\vec{x}_k}.
$$
This sequence also tends towards $\mathcal{B}$ when $n$ goes to infinity.
We now fix $n \in \mathbb{N}.$ Let us show that $\pi$ is a homeomorphism from the restriction of $\mathcal{E}'_n$ to $B(\vec{0},n)$ onto $\pi(B(\vec{0},n))$, that is, $\mathcal{E}'_n$ is ``locally'' a stepped surface. This property 
will be referred to as the ``wanted property'' in the remaining of the proof.
Up to a relabeling, we can assume that the finite sequence $\{\vec{x}_k~|~||\vec{x}_k||\leq n\}$ satisfies:
\begin{equation*}\label{eq:order}
\vec{x}_i\leq \vec{x}_j ~\Rightarrow~ \left\{\begin{array}{ll}
i<j & \textrm{if }\lambda_i >0, \ \lambda_j >0,\\
i>j & \textrm{if }\lambda_i <0, \ \lambda_j <0.
\end{array}\right.
\end{equation*}
Indeed, it suffices to sort $\{\vec{x}_k~|~||\vec{x}_k||\leq n \textrm{ and }\lambda_k>0\}$ and $\{\vec{x}_n~|~||\vec{x}_k||\leq n \textrm{ and }\lambda_k<0\}$ w.r.t. the partial orders $\leq$ and $\geq$, respectively.
Let us show the wanted property by induction on the cardinality $c$  of the set $F_n:=\{\vec{x}_k~|~||\vec{x}_k||\leq n\}$. Note that  $c$ depends on $n$.
This holds for $F_n=\emptyset$, since $\pi$ is a homeomorphism from the whole geometric realization of $\mathcal{S}$ onto $\pi(\mathbb{R}^d)$. 
The next step of the induction consists in adding a weighted pseudo-flip $\lambda_c\mathcal{F}_{\vec{x}_c}$ to a stepped function which satisfies the wanted property.
One checks that the way pseudo-flips have been sorted ensures that, if adding $\lambda_f\mathcal{F}_{\vec{x}_f}$ yields weigth $w<0$ (resp. $w>1$) to some face, then the weight can only be decreased (resp. increased) by further pseudo-flips, \emph{i.e.}, by adding $\lambda_k\mathcal{F}_{\vec{x}_k}$ for $||\vec{x}_k||>n$.
In particular, this contradicts the fact that, after having performed all the pseudo-flips, we get the binary stepped function $\mathcal{B}$.
Hence, adding $\lambda_f\mathcal{F}_{\vec{x}_f}$ simply exchanges the weights ($0$ or $1$) of the faces $\{(\vec{x}_f,i^*),i=1,\ldots,d\}$ (the lower facets of a unit hypercube) and $\{(\vec{x}_f+\vec{e}_i,i^*),i=1,\ldots,d\}$ (the upper facets of the same unit hypercube).
 Note that for any hypercube, the respective images under $\pi$ of the union of its upper facets and of the union of its lower facets. This ensures that, by adding $\lambda_c\mathcal{F}_{\vec{x}_c}$, one gets a stepped function whose restriction to $B(\vec{0},n)$ is still homeomorphic under $\pi$ to $\pi(B(\vec{0},n))$.
This shows that the wanted property holds whatever is the cardinality of $F_n$.\\
Thus, $\mathcal{E}'_n$ is ``locally'' a stepped surface (\emph{i.e.}, on $B(\vec{0},n)$).
In order to finish the proof, let us fix now $r>0$.
Since $\mathcal{B}$ is non-zero, there is $(\vec{x},i)$ such that $\mathcal{B}(\vec{x},i)=1$.
Since the sequence $(\mathcal{E}'_n)_n$ tends towards $\mathcal{B}$, there is $N_r$ such that, for $n\geq N_r$, $\mathcal{E}'_n$ and $\mathcal{B}$ agree on $B(\vec{x},r)$.
Then, for $n\geq r$, $\pi$ is a homeomorphism from the restriction of $\mathcal{E}'_n$ to $B(\vec{x},r)$ -- hence from the restriction of $\mathcal{B}$ to $B(\vec{x},r)$ -- onto $\pi(B(\vec{0},r))$.
When $r$ goes to infinity, this shows that $\pi$ is a homeomorphism from $\mathcal{B}$ onto $\pi(\mathbb{R}^d)$, that is, $\mathcal{B}$ is a stepped surface.
\end{proof}

\noindent Fig.\ \ref{fig:pseudo_flip_acc_reord} illustrates the relabeling (or sorting) of flips in the proof of Prop.\ \ref{prop:binary_pseudo_flip_acc}.

\begin{figure}[hbtp]
\centering
\includegraphics[width=0.6\textwidth]{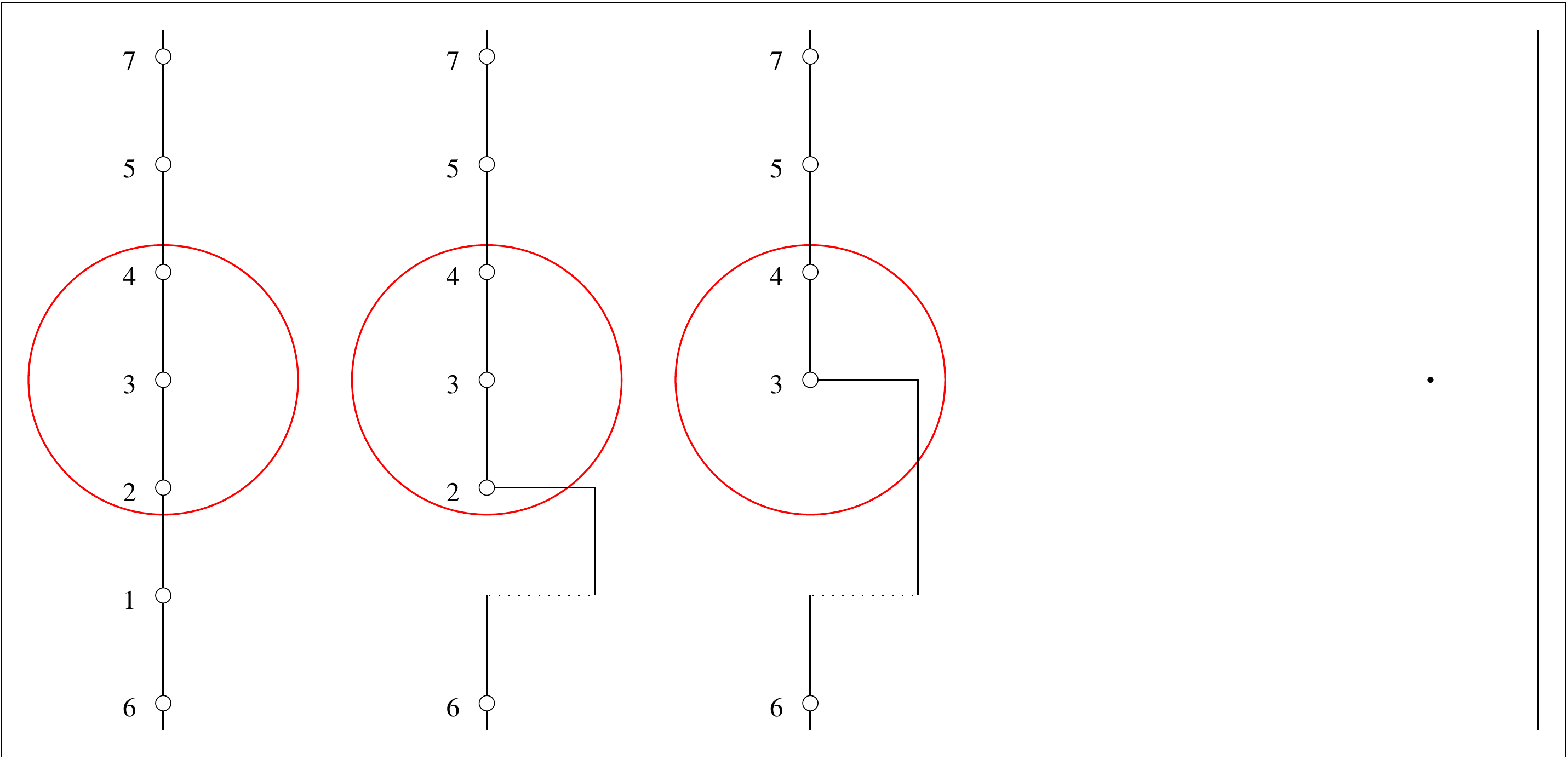}
\caption{
We take the same stepped surfaces $\mathcal{S}$ and $\mathcal{S}'$ as on Fig. \ref{fig:pseudo_flip_acc}.
We consider the open ball $B(\vec{0},1.25)$ (depicted as a circle).
Following the proof of Prop. \ref{prop:binary_pseudo_flip_acc}, we can sort the flips (as depicted on the figure) so that the projection $\pi$ is a homeomorphism from the restrictions to $B(\vec{0},1.25)$ of intermediate stepped functions onto their images under $\pi$.
}
\label{fig:pseudo_flip_acc_reord}
\end{figure}

\section{Action of dual maps}
\label{sec:dual_maps}

In this section, we first recall the notion of \emph{dual map} from \cite{AI,ei} which yields maps defined on stepped functions.
We then characterize their action respectively on stepped planes and surfaces.

\subsection{Dual maps}
\label{sec:def_dual_maps}

We first need to briefly recall some basic notions and notation of words combinatorics.
The free group over the alphabet $\{1,\ldots,d\}$ is denoted by $F_d$.
The identity element is the empty word, usually denoted by $\varepsilon$.
All elements of $F_d$ will be given in their reduced form.
For $w\in F_d$ and a letter $i\in\{1,\ldots,d\}$, one denotes by $|w|_i$ the sum of powers of the occurences of $i$ in $w$.
Then, the \emph{Abelianization map} $\vec{f}$ maps any $w\in F_d$ to the vector of $\mathbb{Z}^d$ whose $i$-th entry is $|w|_i$.
For example:
$$
\vec{f}(12^{-1}1^221^{-1})=(1+2-1,-1+1)=(2,0).
$$
A \emph{morphism} of $F_d$ is a map $\sigma:F_d\to F_d$ such that, for any $u$ and $v$ in $F_d$, $\sigma(uv)=\sigma(u)\sigma(v)$.
It is said to be non-erasing if no letter is mapped onto the empty word.
The \emph{incidence matrix} of a morphism $\sigma$, denoted $M_\sigma$, is the $d\times d$ integer matrix whose $i$-th column is $\vec{f}(\sigma(i))$.
A morphism $\sigma$ is said to be \emph{unimodular} if $\det M_\sigma=\pm 1$, that is, if $M_\sigma$ belongs to the linear group $GL(d,\mathbb{Z})$.
It is said to be \emph{non-negative} if it is   non-negative with respect to the canonical basis (let us stress the fact that the  notion 
 of non-negativity for a  morphism heavily depends on the  choice of    a basis for the free group).
Non-erasing and non-negative morphisms (in the canonical basis) are usually called \emph{substitutions}: they map words with positive power letters onto non-empty words with positive power letters, and thus have non-negative incidence matrix (for more on substitutions, see, \emph{e.g.}, \cite{Pyt}).
For example, $\sigma:1\mapsto 12,~2\mapsto 13,~3\mapsto 1$ is a unimodular substitution with incidence matrix:
$$
M_\sigma=\left(\begin{array}{ccc}
1&1&1\\
1&0&0\\
0&1&0
\end{array}\right).
$$

\noindent We are now in a position to recall the notion of \emph{dual map}:

\begin{definition}\label{def:dual_map}
The \emph{dual map} associated with a unimodular morphism $\sigma$ is the endomorphism of $\mathfrak{F}$ denoted by $E_1^*(\sigma)$ and defined, for any $\mathcal{E}\in\mathfrak{F}$, by:
$$
E_1^*(\sigma)(\mathcal{E}):(\vec{x},i)\mapsto\hspace{-10pt}\sum_{\stackrel{1\leq j\leq d}{\sigma(i)=p\cdotp j\cdotp s}} \hspace{-5pt}\mathcal{E}(M_\sigma\vec{x}+\vec{f}(p),j)-\hspace{-10pt}\sum_{\stackrel{1\leq j\leq d}{\sigma(i)=p\cdotp j^{\textrm{-}1}\cdotp s}}\hspace{-8pt}\mathcal{E}(M_\sigma\vec{x}+\vec{f}(p)-\vec{e}_i,j).
$$
\end{definition}

\noindent Note that $E_1^*(\sigma)(\mathcal{E})$ takes finite values, as finite sums of values of $\mathcal{E}$.
Thus, $E_1^*(\sigma)$ is well-defined.

\begin{remark}
Dual maps have been introduced in \cite{AI} under the name \emph{generalized substitutions} in the particular case where $\sigma$ is a unimodular substitution.
They have been extended to the case of unimodular morphisms of the free group in \cite{ei}.
Extensions to more general spaces than $\mathfrak{F}$ (based on higher codimension facets of cubes) have also been provided in \cite{AIS}.
In particular, we keep here the notation $E_1^*(\sigma)$ introduced in the papers \cite{AI,AIS}, although it can seem at first view obfuscating.
Indeed, the subscript of $E_1^*(\sigma)$ stands for the codimension of the geometric interpretation of faces (in the present paper, they are codimension one facets of hypercubes), while the superscript of $E_1^*(\sigma)$ indicates that it is the dual map of some map $E_1(\sigma)$, hence also the term \emph{dual} map.
\end{remark}

Since elements of $\mathfrak{F}$ are weighted sums of faces and dual maps are linear, the image of any element of $\mathfrak{F}$ can be deduced from the images of faces.
It is thus worth computing the action of dual maps on faces:

\begin{proposition}\label{prop:image_face}
The image under the dual map $E_1^*(\sigma)$ of the face $(\vec{x},i^*)$ is:
$$
E_1^*(\sigma)(\vec{x},i^*)=\hspace{-10pt}\sum_{\stackrel{1\leq j\leq d}{\sigma(j)=p\cdotp i\cdotp s}} \hspace{-5pt}(M_\sigma^{\textrm{-}1}(\vec{x}-\vec{f}(p)),j^*)-\hspace{-10pt}\sum_{\stackrel{1\leq j\leq d}{\sigma(j)=p\cdotp i^{\textrm{-}1}\cdotp s}}\hspace{-8pt}(M_\sigma^{\textrm{-}1}(\vec{x}-\vec{f}(p)+\vec{e}_j),j^*).
$$
\end{proposition}

\begin{proof}
Let $(\vec{y},k)\in\mathbb{Z}^d\times\{1,\ldots,d\}$.
One the one hand, Def.\ \ref{def:dual_map} yields:
$$
E_1^*(\sigma)(\vec{x},i^*)(\vec{y},k)=\hspace{-12pt}\sum_{\stackrel{1\leq j\leq d}{\sigma(k)=p\cdotp j\cdotp s}} \hspace{-5pt}(\vec{x},i^*)(M_\sigma\vec{y}+\vec{f}(p),j)-\hspace{-12pt}\sum_{\stackrel{1\leq j\leq d}{\sigma(k)=p\cdotp j^{\textrm{-}1}\cdotp s}}\hspace{-8pt}(\vec{x},i^*)(M_\sigma\vec{y}+\vec{f}(p)-\vec{e}_k,j).
$$
By definition of a face, $(\vec{x},i^*)(M_\sigma\vec{y}+\vec{f}(p),j)$ is equal either to one if $j=i$ and $M_\sigma\vec{y}+\vec{f}(p)=\vec{x}$, or to zero.
The first sum is thus equal either to one if there is $k$ such that $\sigma(k)=p\cdotp i\cdotp s$ and $M_\sigma\vec{y}+\vec{f}(p)=\vec{x}$, or to zero.\\
One the other hand, consider the formula given in Prop. \ref{prop:image_face}.
Here, $(M_\sigma^{\textrm{-}1}(\vec{x}-\vec{f}(p)),j^*)(\vec{y},k)$ is equal either to one if $j=k$ and $M_\sigma^{\textrm{-}1}(\vec{x}-\vec{f}(p))=\vec{y}$, \emph{i.e.}, $M_\sigma\vec{y}+\vec{f}(p)=\vec{x}$, or to zero.
The first sum of this formula is thus equal either to one if there is $j$ such that $\sigma(j)=p\cdotp i\cdotp s$ and $M_\sigma\vec{y}+\vec{f}(p)=\vec{x}$, or to zero.
This shows that the first sums of both formulaes are identical.
We similarly prove that the second sums of both formulaes -- hence the whole formulaes -- are identical.
This ends the proof.
\end{proof}

It follows from Def. \ref{def:dual_map} that $E_1^*(\sigma)$ has finite range.
Indeed, $E_1^*(\sigma)(\mathcal{E})$ takes on $(\vec{x},i)$ a value which depends only on the values of $\mathcal{E}(\vec{y},j)$ for $||\vec{y}||\leq ||M_\sigma||\times ||\vec{x}||+A_\sigma$, where $A_\sigma$ is a constant depending only on $\sigma$.
Thus, for any $r>0$, if two stepped functions agree on $B(\vec{0},||M_\sigma||r+A_\sigma)$, then their images under $E_1^*(\sigma)$ agree on $B(\vec{0},r)$.
In particular, this shows that dual maps are uniformly continuous on $\mathfrak{F}$.\\

\noindent Let us also mention that, for two unimodular morphisms $\sigma$ and $\sigma'$, one has:
\begin{equation}\label{eq:compo}
E_1^*(\sigma\circ\sigma')=E_1^*(\sigma')\circ E_1^*(\sigma).
\end{equation}
This can be checked by a direct computation from the formula in Def.\ \ref{def:dual_map}, although this also easily follows from the way dual maps are defined in \cite{AI,ei}.

\begin{example}\label{ex:brun_substitutions}
For $a\in\mathbb{N}$ and $1\leq i,j\leq d$, one defines the \emph{Brun substitution}\footnote{The reason for such a name will become clearer in Sec.\ \ref{sec:brun_expansions}.} $\beta_{a,i,j}$ by:
$$
\beta_{a,i,j}(i)=i^aj,\quad\beta_{a,i,j}(j)=i,\quad\forall k\notin\{i,j\},~\beta_{a,i,j}(k)=k.
$$
Its incidence matrix is
$$
\left(\begin{array}{ccccc}
I_{i-1} & & & &\\
 & a & & 1 &\\
 & & I_{j-i-1} & &\\
 & 1 & & 0 &\\
 & & & & I_{d-j}
\end{array}\right).
$$
One checks that $\beta_{a,i,j}$ is unimodular.
It is also invertible, with inverse:
$$
\beta_{a,i,j}^{-1}(i)=j,\quad\beta_{a,i,j}^{-1}(j)=j^{-a}i,\quad\forall k\notin\{i,j\},~\beta_{a,i,j}(k)^{-1}=k.
$$
Let us compute the images under $E_1^*(\beta_{a,i,j})$ and $E_1^*(\beta_{a,i,j}^{-1})$ of the faces located at $\vec{0}$:
$$
E_1^*(\beta_{a,i,j})~:~\left\{\begin{array}{lll}
(\vec{0},i^*)&\mapsto&(\vec{0},j^*)+\sum_{k=0}^{a-1}(-k\vec{e}_j,i^*),\\
(\vec{0},j^*)&\mapsto&(-a\vec{e}_j,i^*),\\
(\vec{0},k^*)&\mapsto&(\vec{0},k^*).
\end{array}\right.
$$
$$
E_1^*(\beta_{a,i,j}^{-1})~:~\left\{\begin{array}{lll}
(\vec{0},i^*) & \mapsto & (a\vec{e}_i,j^*),\\
(\vec{0},j^*) & \mapsto & (\vec{0},i^*)-\sum_{k=0}^{a-1}(k\vec{e}_i,j^*),\\
(\vec{0},k^*) & \mapsto & (\vec{0},k^*).
\end{array}\right.
$$
Fig.\ \ref{fig:brun_substitutions} illustrates the action of $E_1^*(\beta_{a,i,j})$ for $a=2$, $i=3$, $j=1$.
\end{example}

\begin{figure}[hbtp]
\centering
\includegraphics[width=0.8\textwidth]{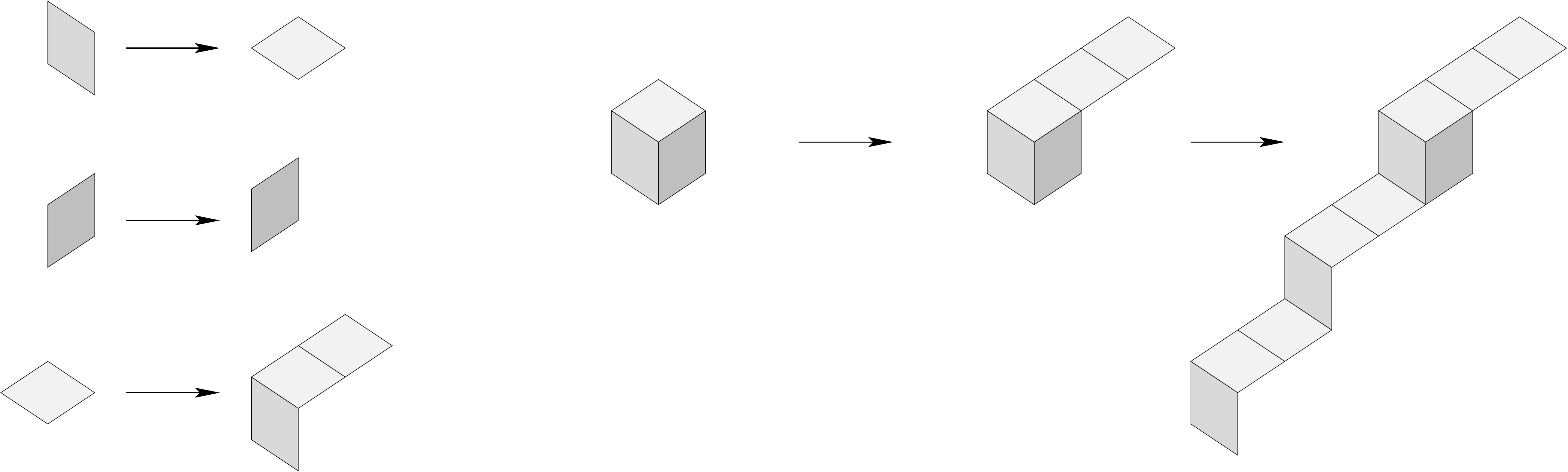}
\caption{
Action of the dual map $E_1^*(\beta_{2,3,1})$ on each of the three types of faces (left).
Two iterations of $E_1^*(\beta_{2,3,1})$ on the upper faces of a unit cube (right).
}
\label{fig:brun_substitutions}
\end{figure}

\subsection{Action on stepped planes}
\label{sec:sp}

In the previous section, dual maps have been defined over $\mathfrak{F}$.
Here, we are interested in their action over $\mathfrak{P}$, which is characterized by the main result of this section:

\begin{theorem}\label{th:image_stepped_plane}
Let $\sigma$ be a unimodular morphism.
Let $\vec{\alpha}\in\mathbb{R}_+^d\backslash\{\vec{0}\}$ and $\rho\in\mathbb{R}$.
Then, the action of $E_1^*(\sigma)$ on the stepped plane $\mathcal{P}_{\vec{\alpha},\rho}$ satisfies the equivalence:
$$
E_1^*(\sigma)(\mathcal{P}_{\vec{\alpha},\rho})\in\mathfrak{B}
~\Leftrightarrow~
E_1^*(\sigma)(\mathcal{P}_{\vec{\alpha},\rho})=\mathcal{P}_{M_\sigma^\top\vec{\alpha},\rho}
~\Leftrightarrow~
M_\sigma^\top\vec{\alpha}\in\mathbb{R}_+^d.
$$
\end{theorem}

\begin{proof}
Let $\mathcal{E}=E_1^*(\sigma)(\mathcal{P}_{\vec{\alpha},\rho})$.
Let us first compute $\mathcal{E}(\vec{y},j)$, with $(\vec{y},j)\in\mathbb{Z}^d\times\{1,\ldots,d\}$ being fixed.
We write $\sigma(j)=u_1^{\varepsilon_1}\ldots u_l^{\varepsilon_l}$, where, for $k=1,\ldots,l$, $u_k\in\{1,\ldots,d\}$ and $\varepsilon_k=\pm 1$.
Each $(\vec{x},i)\in\mathbb{Z}^d\times\{1,\ldots,d\}$ such that $\mathcal{P}_{\vec{\alpha},\rho}(\vec{x},i)=1$ contributes to:\\
$\bullet$ increment $\mathcal{E}(\vec{y},j)$ by $1$ for each $k\in\{1,\ldots,l\}$ such that $u_k^{\varepsilon_k}=i$ and $\vec{y}=M_\sigma^{-1}(\vec{x}-\vec{f}(u_1^{\varepsilon_1}\ldots u_{k-1}^{\varepsilon_{k-1}}))$;\\
$\bullet$ decrement $\mathcal{E}(\vec{y},j)$ by $1$ for each $k\in\{1,\ldots,l\}$ such that $u_k^{\varepsilon_k}=i^{-1}$ and $\vec{y}=M_\sigma^{-1}(\vec{x}-\vec{f}(u_1^{\varepsilon_1}\ldots u_{k-1}^{\varepsilon_{k-1}})+\vec{e}_i)$.\\
By introducing $\vec{r}_k=M_\sigma\vec{y}+\vec{f}(u_1^{\varepsilon_1}\ldots u_k^{\varepsilon_k})$, for $0\leq k\leq l$, the first case becomes:
$$
\vec{x}=\vec{r}_{k-1}
\quad\textrm{and}\quad
\vec{x}+\vec{e}_i=\vec{r}_{k-1}+\vec{e}_i=\vec{r}_{k-1}+\vec{f}(i)=\vec{r}_{k-1}+\vec{f}(u_k^{\varepsilon_k})=\vec{r}_k,
$$
and the second:
$$
\vec{x}=\vec{r}_{k-1}-\vec{e}_i=\vec{r}_{k-1}+\vec{f}(i^{-1})=\vec{r}_{k-1}+\vec{f}(u_k^{\varepsilon_k})=\vec{r}_k
\quad\textrm{and}\quad
\vec{x}+\vec{e}_i=\vec{r}_{k-1}.
$$
Since $\mathcal{P}_{\vec{\alpha},\rho}(\vec{x},i)=1$ if and only if $\langle \vec{x}|\vec{\alpha}\rangle<\rho\leq \langle \vec{x}+\vec{e}_i|\vec{\alpha}\rangle$, one computes:
$$
\mathcal{E}(\vec{y},j)=\#\{k~|~\langle \vec{r}_{k-1}|\vec{\alpha}\rangle<\rho\leq \langle \vec{r}_k|\vec{\alpha}\rangle\}-\#\{k~|~\langle \vec{r}_k|\vec{\alpha}\rangle<\rho\leq \langle \vec{r}_{k-1}|\vec{\alpha}\rangle\}.
$$
One then gets (see Fig.\ \ref{fig:proof_image_stepped_plane}):
$$
\mathcal{E}(\vec{y},j)=\left\{\begin{array}{rl}
1 & \textrm{ if }\langle \vec{r}_0|\vec{\alpha}\rangle<\rho\leq \langle \vec{r}_l|\vec{\alpha}\rangle,\\
-1 & \textrm{ if }\langle \vec{r}_l|\vec{\alpha}\rangle<\rho\leq \langle \vec{r}_0|\vec{\alpha}\rangle,\\
0 & \textrm{ otherwise.}
\end{array}\right.
$$
Last, note that one has:
$$
\langle \vec{r}_0|\vec{\alpha}\rangle=\langle M_\sigma\vec{y}|\vec{\alpha}\rangle=\langle \vec{y}|M_\sigma^\top\vec{\alpha}\rangle,
$$
$$
\langle \vec{r}_l|\vec{\alpha}\rangle=\langle M_\sigma\vec{y}+\vec{f}(\sigma(i))|\vec{\alpha}\rangle=\langle M_\sigma(\vec{y}+\vec{e}_i)|\vec{\alpha}\rangle=\langle \vec{y}+\vec{e}_i|M_\sigma^\top\vec{\alpha}\rangle.
$$
We can now rely on this computation to prove the claimed equivalences.
If $M_\sigma^\top\vec{\alpha}\in\mathbb{R}_+^d$, then $\langle \vec{r}_0|\vec{\alpha}\rangle\leq \langle \vec{r}_l|\vec{\alpha}\rangle$, and above computations yield that $\mathcal{E}(\vec{y},j)=1$ if and only if $\langle \vec{y}|M_\sigma^\top\vec{\alpha}\rangle <\rho\leq \langle \vec{y}+\vec{e}_i|M_\sigma^\top\vec{\alpha}\rangle$, \emph{i.e.}, $\mathcal{E}=\mathcal{P}_{M_\sigma^\top\vec{\alpha},\rho}$
according to Def.\ \ref{def:stepped_plane}.
On the contrary, $M_\sigma^\top\vec{\alpha}\notin\mathbb{R}_+^d$ yields that there is $i$ such that $\langle \vec{r}_0|\vec{\alpha}\rangle> \langle \vec{r}_l|\vec{\alpha}\rangle$, and above computations show that there is a face $(\vec{y},j^*)$ with $\mathcal{E}(\vec{y},j)=-1$, hence $\mathcal{E}\notin\mathfrak{B}$.
Other implications are straightforward.
\end{proof}

\begin{figure}[hbtp]
\centering
\includegraphics[width=\textwidth]{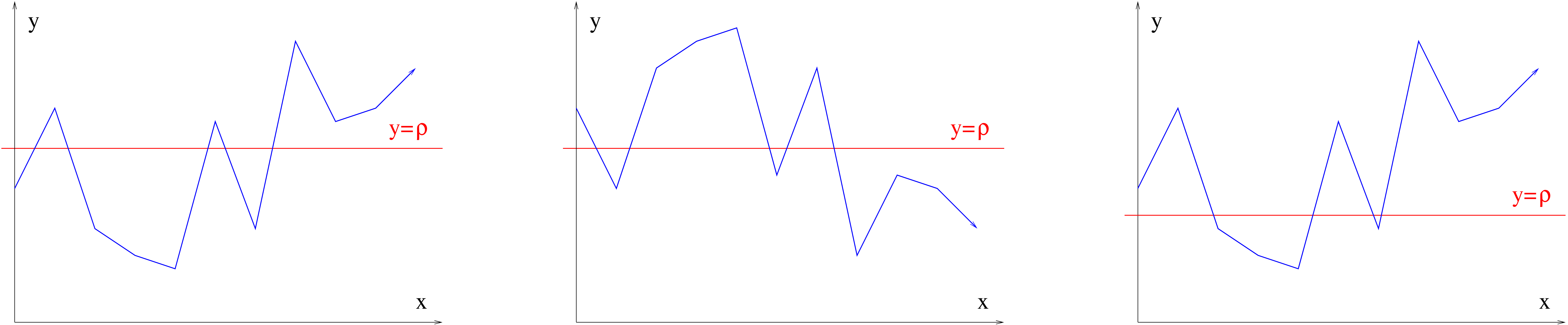}
\caption{Consider the broken line joining points $(x=k,y=\langle\vec{r}_k|\vec{\alpha}\rangle)$, for $k=0,\ldots,l$.
The value $\mathcal{E}(\vec{y},j)$ is equal to the number of times this broken line crosses upwards the line $y=\rho$ minus the number of times it crosses downwards this line.
It is thus equal to $1$ if the first point is below this line and the last one above, $-1$ if the first point is above this line and the last one below, and $0$ otherwise (from left to right).
}
\label{fig:proof_image_stepped_plane}
\end{figure}

Fig.\ \ref{fig:image_stepped_plane} gives an example how a dual map acts on a stepped plane.
Note that the previous theorem in particular holds if $M_\sigma$ is non-negative: in such a case, the image of any stepped plane is a stepped plane.\\

\begin{figure}[hbtp]
\centering
\includegraphics[width=\textwidth]{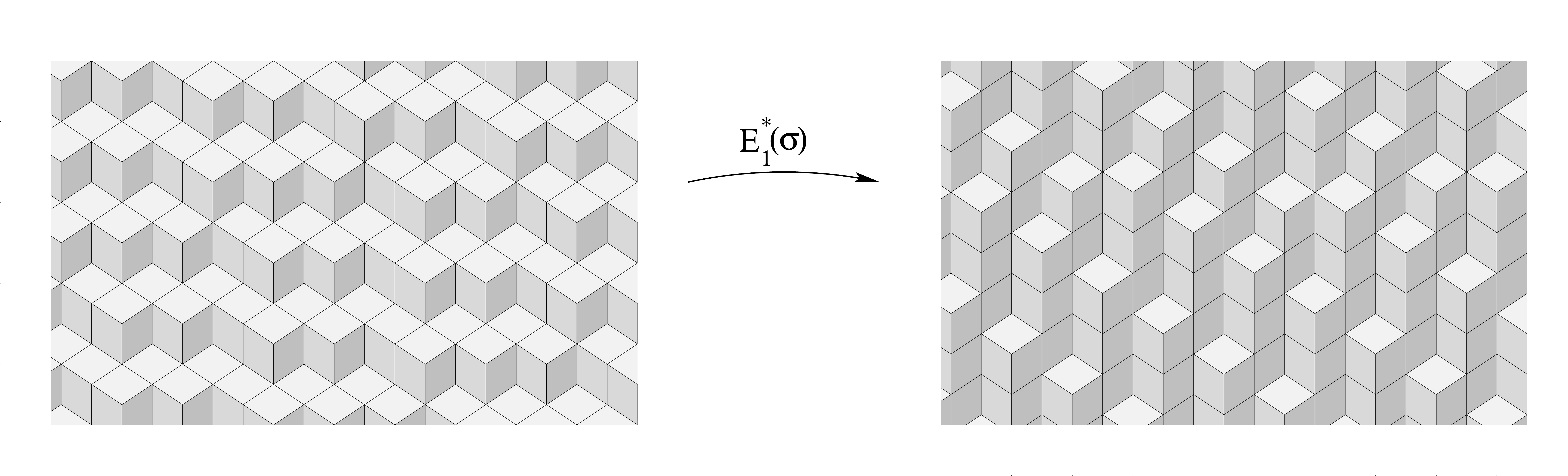}
\caption{
If the image of a stepped plane $\mathcal{P}_{\vec{\alpha},\rho}$ under a dual map $E_1^*(\sigma)$ is a binary stepped function, then it is the stepped plane $\mathcal{P}_{M_\sigma^\top\vec{\alpha},\rho}$.
}
\label{fig:image_stepped_plane}
\end{figure}

\begin{remark}\label{rem:digitization_linear_group1}
For any invertible matrix $M$, one has 
$$
M^{-1}\{\vec{x}\in\mathbb{R}^d~|~\langle\vec{x}|\vec{\alpha}\rangle=\rho\}=\{\vec{x}\in\mathbb{R}^d~|~\langle\vec{x}|M^\top\vec{\alpha}\rangle=\rho\}.
$$
This holds in particular for the incidence matrix $M_{\sigma}$ 
of any unimodular morphism $\sigma$.
In other words, $M_\sigma^{-1}$ maps a real plane of normal vector $\vec{\alpha}$ onto a real plane of normal vector $M_\sigma^\top\vec{\alpha}$.
By comparing with the action of the dual map $E_1^*(\sigma)$ on stepped planes and by noting that stepped planes are digitizations of real planes, we get a more intuitive viewpoint: the action of dual maps over stepped planes can be seen as a digitization of the action of the linear group $GL(d,{\mathbb Z})$ over real planes.
\end{remark}

\begin{remark}\label{rem:commutative_action1}
According to Th.\ \ref{th:image_stepped_plane}, the image of a stepped plane under a dual map $E_1^*(\sigma)$ depends only on the incidence matrix of $\sigma$.
This enables to make choice by picking a morphism whose incidence matrix is given (we will rely on this in Sec. \ref{sec:brun_expansions} to define suitable \emph{Brun substitutions}).
This is however generally false for a stepped function, for example, a single face.
Nevertheless, this does not contradict the fact that the action of $\sigma$ on the free group is non-commutative and not characterized by $M_\sigma$.
\end{remark}

\noindent To conclude this section, let us stress two weakened statements of Th.\ \ref{th:image_stepped_plane}:

\begin{corollary}\label{cor:image_stepped_plane}
Let $\sigma$ be a unimodular morphism.
If the image of a stepped plane under $E_1^*(\sigma)$ is a binary stepped function, then it is a stepped plane:
$$
E_1^*(\sigma)(\mathfrak{P})\cap\mathfrak{B}\subset\mathfrak{P}.
$$
\end{corollary}

\begin{corollary}\label{cor:positive_image_stepped_plane}
Let $\sigma$ be a unimodular morphism.
If the incidence matrix of $\sigma$ is non-negative, then $E_1^*(\sigma)$ maps any stepped plane onto a stepped plane:
$$
M_\sigma\geq 0 ~\Rightarrow~ E_1^*(\sigma)(\mathfrak{P})\subset\mathfrak{P}.
$$
\end{corollary}

\subsection{Action on stepped surfaces}
\label{sec:ss}

In this section, the notion of flip (introduced in Sec.\ \ref{sec:flips}) is used to extend the results of the previous section from stepped planes to stepped surfaces.
The idea is to see stepped surfaces as flips performed on stepped planes.
We first characterize the action of dual maps on flips:

\begin{proposition}\label{prop:image_flip}
Let $\sigma$ be a unimodular morphism and $\vec{x}\in\mathbb{Z}^d$.
Then, $E_1^*(\sigma)$ maps the flip located at $\vec{x}$ onto the flip located at $M_\sigma^{-1}\vec{x}$:
$$
E_1^*(\sigma)(\mathcal{F}_{\vec{x}})=\mathcal{F}_{M_\sigma^{-1}\vec{x}}.
$$
\end{proposition}

\begin{proof}
Let $\mathcal{E}=E_1^*(\sigma)(\mathcal{F}_{\vec{x}})$.
Fix $(\vec{y},j)\in\mathbb{Z}^d\times\{1,\ldots,d\}$ and compute $\mathcal{E}(\vec{y},j)$.
As in the proof of Th.\ \ref{th:image_stepped_plane}, let us write $\sigma(j)=u_1^{\varepsilon_1}\ldots u_l^{\varepsilon_l}$, where, for $k=1,\ldots,l$, $u_k\in\{1,\ldots,d\}$ and $\varepsilon_k=\pm 1$, and let us introduce, for $k=0,\ldots,l$:
$$
\vec{r}_k=M_\sigma\vec{y}+\vec{f}(u_1^{\varepsilon_1}\cdots u_k^{\varepsilon_k}).
$$
First, Def.\ \ref{def:dual_map} yields:
$$
\mathcal{E}(\vec{y},j)=\sum_{\varepsilon_k>0} \mathcal{F}_{\vec{x}}(\vec{r}_{k-1},u_k)-\sum_{\varepsilon_k<0}\mathcal{F}_{\vec{x}}(\underbrace{\vec{r}_{k-1}-\vec{e}_{u_k}}_{=\vec{r}_k},u_k).
$$
Then, Def.\ \ref{def:flip} yields both:
$$
\sum_{\varepsilon_k>0} \mathcal{F}_{\vec{x}}(\vec{r}_{k-1},u_k)=\#\{\varepsilon_k>0~|~\vec{r}_{k-1}=\vec{x}\}-\#\{\varepsilon_k>0~|~\underbrace{\vec{r}_{k-1}=\vec{x}-\vec{e}_{u_k}}_{\Leftrightarrow~\vec{r}_k=\vec{x}}\}
$$
and:
$$
\sum_{\varepsilon_k<0}\mathcal{F}_{\vec{x}}(\vec{r}_k,u_k)=\#\{\varepsilon_k<0~|~\vec{r}_k=\vec{x}\}-\#\{\varepsilon_k<0~|~\underbrace{\vec{r}_k=\vec{x}-\vec{e}_{u_k}}_{\Leftrightarrow~\vec{r}_{k-1}=\vec{x}}\}.
$$
One deduces:
$$
\mathcal{E}(\vec{y},j)
=\#\{\varepsilon_k~|~\vec{r}_{k-1}=\vec{x}\}-\#\{\varepsilon_k~|~\vec{r}_k=\vec{x}\}
=\left\{\begin{array}{rl}
1 & \textrm{ if }\vec{x}=\vec{r}_0,\\
-1 & \textrm{ if }\vec{x}=\vec{r}_l,\\
0 & \textrm{ otherwise.}
\end{array}\right.
$$
Since $\vec{r}_0=M_\sigma\vec{y}$ and $\vec{r}_l=M_\sigma\vec{y}+\vec{f}(\sigma(j))=M_\sigma(\vec{y}+\vec{e}_j)$, the above equation shows that $\mathcal{E}(\vec{y},j)=\mathcal{F}_{M_\sigma^{-1}\vec{x}}(\vec{y},j)$.
The result follows.
\end{proof}

\noindent We also need a technical lemma which plays the role of the $M_\sigma^\top\vec{\alpha}\in\mathbb{R}_+^d$ condition of Th.\ \ref{th:image_stepped_plane}:

\begin{lemma}\label{lem:positivement_orientable}
Let $\sigma$ be a unimodular morphism of the free group.
If there is a stepped surface whose image under $E_1^*(\sigma)$ is a binary stepped function, then there is $\vec{\alpha}\in\mathbb{R}_+^d\backslash\{\vec{0}\}$ such that $M_\sigma^\top\vec{\alpha}\in\mathbb{R}_+^d$:
$$
E_1^*(\sigma)(\mathfrak{S})\cap\mathfrak{B}\neq\emptyset ~\Rightarrow~ \mathbb{R}_+^d\cap M_\sigma^\top\mathbb{R}_+^d\neq \{\vec{0}\}.
$$
\end{lemma}

\begin{proof}
For any stepped function $\mathcal{E}$ of finite size, let $\vec{g}(\mathcal{E})$ be the vector of $\mathbb{Z}^d$ whose $i$-th entry is equal to the sum of the values of $\mathcal{E}$ over $\mathbb{Z}^d\times\{i\}$ ($\vec{g}$ plays on $\mathfrak{F}$ the role played on $F_d$ by the Abelianization map $\vec{f}$).
According to Prop.\ \ref{prop:image_face}, one has:
$$
\vec{g}(E_1^*(\sigma)(\mathcal{E}))=M_\sigma^\top\vec{g}(\mathcal{E}).
$$
Now, assume that $\mathcal{S}$ is a stepped surface such that $E_1^*(\sigma)(\mathcal{S})\in\mathfrak{B}$.
For $n\in\mathbb{N}$, the restriction of $\mathcal{S}$ to $B(\vec{0},n)$ is a finite binary stepped function that we denote by $\mathcal{S}_n$.
The vector $\vec{g}(\mathcal{S}_n)$ is non-negative; it is also non-zero for $n$ big enough.
Since $M_\sigma^\top$ is invertible, $\vec{g}(E_1^*(\sigma)(\mathcal{S}_n))=M_\sigma^\top\vec{g}(\mathcal{S}_n)$ is also non-zero.
We can thus define:
$$
\vec{\alpha}_n=\frac{\vec{g}(\mathcal{S}_n)}{||\vec{g}(\mathcal{S}_n)||}
\quad\textrm{and}\quad
\vec{\beta}_n=\frac{\vec{g}(E_1^*(\sigma)(\mathcal{S}_n))}{||\vec{g}(E_1^*(\sigma)(\mathcal{S}_n))||}.
$$
These two vectors have norm $1$ and satisfy $M_\sigma^\top\vec{\alpha}_n=\lambda_n\vec{\beta}_n$, where:
$$
0\leq \lambda_n=\frac{||\vec{g}(E_1^*(\sigma)(\mathcal{S}_n))||}{||\vec{g}(\mathcal{S}_n)||}=\frac{||M_\sigma^\top\vec{g}(\mathcal{S}_n)||}{||\vec{g}(\mathcal{S}_n)||}\leq ||M_\sigma^\top||.
$$
By compactness, we can extract convergent sequences from the bounded sequences $(\vec{\alpha}_n)$, $(\vec{\beta}_n)$ and $(\lambda_n)$.
Let $\vec{\alpha}$, $\vec{\beta}$ and $\lambda$ denote their respective limits: they are non-negative and satisfy $M_\sigma^\top\vec{\alpha}=\lambda\vec{\beta}$.
Moreover, $\vec{\alpha}$ is non-zero since the $\vec{\alpha}_n$ have norm $1$.
This ends the proof.
\end{proof}

We are now in a position to extend Th.\ \ref{th:image_stepped_plane} (more exactly, Cor.\ \ref{cor:image_stepped_plane}) to the stepped surface case, by combining all the results previously obtained:

\begin{theorem}\label{th:image_stepped_surface}
Let $\sigma$ be a unimodular morphism.
If the image of a stepped surface under $E_1^*(\sigma)$ is a binary stepped function, then it is a stepped surface:
$$
E_1^*(\sigma)(\mathfrak{S})\cap\mathfrak{B}\subset\mathfrak{S}.
$$
\end{theorem}

\begin{proof}
Assume that $\mathcal{S}$ is a stepped surface such that $E_1^*(\sigma)(\mathcal{S})\in\mathfrak{B}$.
Lem.\ \ref{lem:positivement_orientable} ensures that there is $\vec{\alpha}\in\mathbb{R}_+^d\backslash\{\vec{0}\}$ such that $M_\sigma^\top\vec{\alpha} \in\mathbb{R}_+^d$.
Since $\vec{\alpha}\in\mathbb{R}_+^d\backslash\{\vec{0}\}$,
Prop.\ \ref{prop:pseudo_flip_acc} thus yields that $\mathcal{S}$ is pseudo-flip-accessible from the stepped plane $\mathcal{P}_{\vec{\alpha},0}$:
$$
\mathcal{S}=\mathcal{P}_{\vec{\alpha},0}+\lim_{n\to\infty}\sum_{k\leq n}\varepsilon_k \mathcal{F}_{\vec{x}_k},
$$
where, for any $k$, $\varepsilon_k=\pm 1$ and $\vec{x}_k\in\mathbb{Z}^d$.
The uniform continuity of $E_1^*(\sigma)$ gives:
$$
E_1^*(\sigma)(\mathcal{S})=E_1^*(\sigma)(\mathcal{P}_{\vec{\alpha},0})+\lim_{n\to\infty}\sum_{k\leq n}\varepsilon_k E_1^*(\sigma)(\mathcal{F}_{\vec{x}_k}).
$$
Then, Th.\ \ref{th:image_stepped_plane} and Prop.\ \ref{prop:image_flip} imply:
$$
E_1^*(\sigma)(\mathcal{S})=\mathcal{P}_{M_\sigma^\top\vec{\alpha},0}+\lim_{n\to\infty}\sum_{k\leq n}\varepsilon_k \mathcal{F}_{M_\sigma^{-1}\vec{x}_k}.
$$
This shows that the binary stepped function $E_1^*(\sigma)(\mathcal{S})$ is pseudo-flip-accessible from $\mathcal{P}_{M_\sigma^\top\vec{\alpha},0}$.
Prop.\ \ref{prop:binary_pseudo_flip_acc} thus finally yields that $E_1^*(\sigma)(\mathcal{S})$ is a stepped surface.
\end{proof}

\noindent Fig.\ \ref{fig:image_stepped_surface} illustrates this theorem. 
Note that if $M_{\sigma}$ is non-negative, then 
 the image by $E_1^* (\sigma)$ of any stepped surface is a stepped surface.
\begin{figure}[hbtp]
\centering
\includegraphics[width=\textwidth]{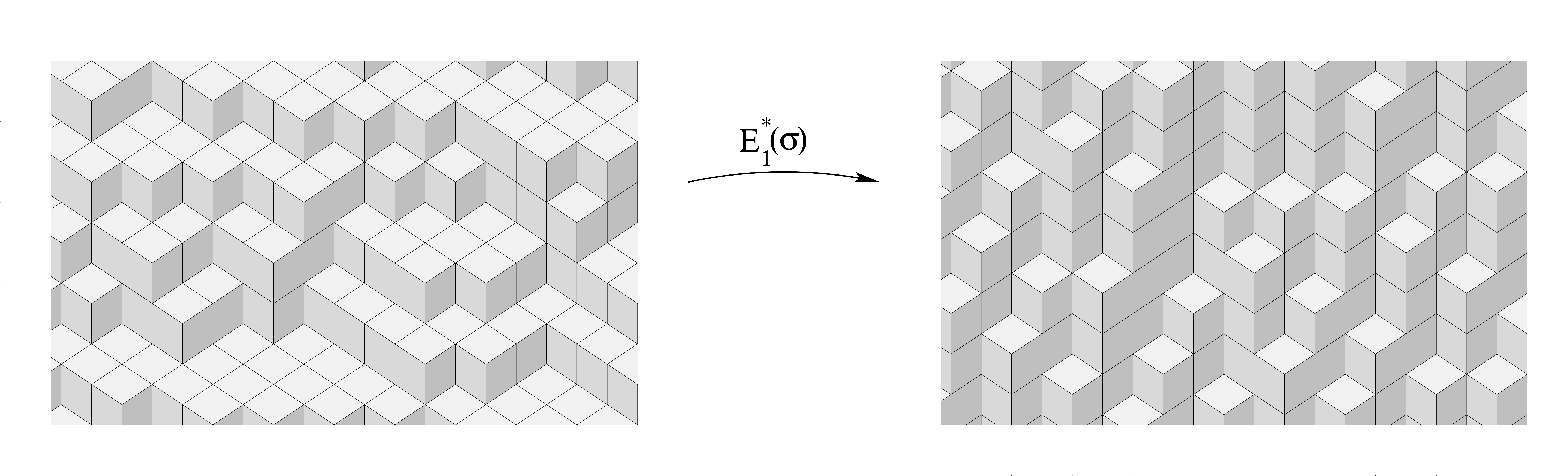}
\caption{
If the image of a stepped surface under a dual map $E_1^*(\sigma)$ is a binary stepped function, then it is a stepped surface.
}
\label{fig:image_stepped_surface}
\end{figure}

\begin{remark}\label{rem:digitization_linear_group2}
The viewpoint according to which the action of dual maps over stepped planes is a digitization of the action of the linear group over real planes (Rem.\ \ref{rem:digitization_linear_group1}) can be extended to stepped surfaces.
Indeed, by seeing a stepped surface as flips added to a stepped plane (according to Prop.\ \ref{prop:pseudo_flip_acc}), the action of a dual map is decomposed into its action on a stepped plane and on flips, with each of them being moved according to the linear map $M_\sigma^{-1}$ (according to Prop.\ \ref{prop:image_flip}).
\end{remark}

As for stepped planes (Rem.\ \ref{rem:commutative_action1}), the action of dual maps over stepped surfaces only depends on incidence matrices.
Indeed, since it holds both for stepped planes and flips (see Prop.\ \ref{prop:image_flip}), Prop.\ \ref{prop:pseudo_flip_acc} easily yields the following:

\begin{theorem}\label{theo:commutative_action2}
Let $\sigma$, $\sigma'$ be two unimodular morphisms
with the same incidence matrix. If the image of of a stepped surface $\mathcal{S}$
 by $E_1^* (\sigma)$ is a binary stepped function, then its images by 
$E_1^* (\sigma)$ and $E_1^* (\sigma')$ do coincide.
\end{theorem}

To conclude this section, let us show that the action of dual maps over stepped planes in the substitution case (Cor.\ \ref{cor:positive_image_stepped_plane}) admits a natural extension for stepped surfaces:

\begin{proposition}\label{prop:positive_image_stepped_surface}
Let $\sigma$ be a unimodular morphism.
If the incidence matrix of $\sigma$ is non-negative, then $E_1^*(\sigma)$ maps any stepped surface onto a stepped surface:
$$
M_\sigma\geq 0 ~\Rightarrow~ E_1^*(\sigma)(\mathfrak{S})\subset\mathfrak{S}.
$$
\end{proposition}

\begin{proof}
Since the action of dual maps over stepped surfaces depends only on incidences matrices, we can w.l.o.g. assume that $\sigma$ is non-negative.
Let us then show that $\sigma\geq 0$ and $E_1^*(\sigma)(\mathfrak{S})\not\subset\mathfrak{S}$ yield a contradiction.
On the one hand, $E_1^*(\sigma)(\mathfrak{S})\not\subset\mathfrak{S}$ provides us $\mathcal{S}\in\mathfrak{S}$ such that $E_1^*(\sigma)(\mathcal{S})\notin\mathfrak{S}$.
One even deduces from Th.\ \ref{th:image_stepped_surface} that $E_1^*(\sigma)(\mathcal{S})\notin\mathfrak{B}$.
On the other hand, $\sigma\geq 0$ yields that the weights of faces in $E_1^*(\sigma)(\mathcal{S})$ are non-negative, according to the expression
of $E_1^*(\sigma)$ provided by Prop.~\ref{prop:image_face}.
Thus, $E_1^*(\sigma)(\mathcal{S})\notin\mathfrak{B}$ means that there are two distinct faces $(\vec{x}_1,i_1^*)$ and $(\vec{x}_2,i_2^*)$ with weight $1$
in $\mathcal{S}$ whose images under $E_1^*(\sigma)$ ``overlap'', \emph{i.e.}, for some $j\in\{1,\ldots,d\}$:
$$
\sigma(j)=p_1\cdotp i_1\cdotp s_1 = p_2\cdotp i_2\cdotp s_2
\quad\textrm{and}\quad
M_\sigma^{-1}(\vec{x}_1-\vec{f}(p_1))=M_\sigma^{-1}(\vec{x}_2-\vec{f}(p_2)).
$$
Note that $(\vec{x}_1,i_1^*)\neq(\vec{x}_2,i_2^*)$ yields $p_1\neq p_2$.
Assume, w.l.o.g., that $p_2$ is shorter than $p_1$.
One computes:
$$
\vec{x}_1-\vec{x}_2=\vec{f}(p_1)-\vec{f}(p_2)\geq \vec{f}(i_2)=\vec{e}_{i_2}.
$$
Let us now introduce the two following integer vectors:
$$
\vec{y}_1=\vec{x}_1+\sum_{i=1}^d\vec{e}_i
\quad\textrm{and}\quad
\vec{y}_2=\vec{x}_2+\vec{e}_{i_2}.
$$
They respectively belong to the geometric interpretation of $(\vec{x}_1,i_1^*)$ and $(\vec{x}_2,i_2^*)$.
Moreover, one has:
$$
\vec{y}_1-\vec{y}_2=(\vec{x}_1+\sum_{i=1}^d\vec{e}_i)-(\vec{x}_2+\vec{e}_{i_2}) = (\underbrace{\vec{x}_1-\vec{x}_2-\vec{e}_{i_2}}_{\geq \vec{0}}) + \sum_{i=1}^d \vec{e}_i>\vec{0}.
$$
The wanted contradiction thus follows from Prop.\ \ref{prop:stepped_surface}.
\end{proof}

\section{Brun expansions}
\label{sec:brun_expansions}
This section provides an interpretation of multi-dimensional continued fraction algorithms in terms of dual maps acting on stepped planes or surfaces.

\subsection{Multi-dimensional continued fractions}
\label{sec:exp_vector}

In the one-dimensional case, there is a canonical continued fraction algorithm, namely the Euclidean algorithm.
In the multi-dimensional case, there is no such canonical algorithm, and several different definitions have been proposed (see Brentjes \cite{brentjes} or Schweiger \cite{schweiger} for a summary).
Here, we follow the definition of Lagarias \cite{lagarias}, where multi-dimensional continued fraction algorithms produce sequences of matrices in $GL(d,\mathbb{Z})$.
We first state the general definition of such algorithms, and then detail one of them, namely the Brun algorithm.

\begin{definition}\label{def:continued_fraction}
Let $X\subset\mathbb{R}^d$ and $X_0\subset X$.
Elements of $X_0$ are called \emph{terminal}.
A $d$-dimensional \emph{continued fraction map} over $X$ is a map $T:X\to X$ such that $T(X_0)\subset X_0$ and, for any $\vec{\alpha}\in X$, there is $B(\vec{\alpha})\in GL(d,\mathbb{Z})$ satisfying:
$$
\vec{\alpha}=B(\vec{\alpha}) \, T(\vec{\alpha}).
$$
The associated \emph{continued fraction algorithm} consists in iteratively applying the map $T$ on a vector $\vec{\alpha}\in X$.
This yields the following sequence of matrices, called the \emph{continued fraction expansion} of $\vec{\alpha}$:
$$
(B(T^n(\vec{\alpha})))_{n\geq 1}.
$$
This expansion is said to be \emph{finite} if there is $n$ such that $T^n(\vec{\alpha})\in X_0$, infinite otherwise.
In the former case, the smallest such $n$ is called the \emph{length} of the expansion.
\end{definition}

\noindent The Jacobi-Perron, Poincar\'e, Selmer or Brun algorithms match this definition.\\

The intuitive principle of a continued fraction algorithm is that each application of the map $T$ captures a partial information about $\vec{\alpha}$, with all the information being captured by the whole expansion.
A good algorithm should capture as much as possible information with as less as possible matrices.
This can be formalized by various notions of \emph{convergence} for continued fraction algorithms.
Here, let us just recall \emph{weak convergence}.\\

We say that a sequence $(\vec{y}_n)$ of non-zero vectors tends \emph{in direction} towards a non-zero vector $\vec{z}$ if $(\vec{y}_n/||\vec{y}_n||)_n$ tends towards $\vec{z}/||\vec{z}||$, that is
$$
\lim_{n \rightarrow \infty} d(\vec{y}_n/||\vec{y}_n||, \mathbb{R} \vec{z})=0.
$$

\begin{definition}\label{def:weak_cv}
A muldi-dimensional continued fraction algorithm is said to be \emph{weakly convergent} in $\vec{\alpha}\in X$ if the expansion $(B_n)_{n\geq 1}$ of $\vec{\alpha}$ is such that the sequence $( B_1\ldots B_n\vec{x})_n$ tends in direction towards $\vec{\alpha}$, uniformly for $\vec{x}\in X$.
\end{definition}

In particular, weakly convergent continued fraction algorithms allow to approximate real vectors by sequences of rational vectors.
Indeed, let $\vec{\alpha}$ be a $d-1$-dimensional real vector of $X$ for $d\geq 2$.
Let then $(B_n)_n$ be the continued fraction expansion of the $d$-dimensional vector $(1,\vec{\alpha})$, that is, for all $n\geq 1$, $B_n= (B(T^n(\vec{\alpha})))$, and define:
$$
(q_n,\vec{p}_n)=B_1 B_2 \ldots B_n(1,\vec{0}),
$$
where $(1,\vec{0})$ is a $d$-dimensional vector and $\vec{p}_n$ is a $(d-1)$-dimensional vector.
The sequence $(\vec{p}_n/q_n)_n$ is thus a sequence of $(d-1)$-dimensional rational vectors, and the weak convergence ensures that it tends towards $\vec{\alpha}$.\\

Let us end this section by presenting the \emph{Brun algorithm}, one of the most classical multi-dimensional continued fraction algorithms matching Def.\ \ref{def:continued_fraction} (see \cite{an} and \cite{brun}).
In this case, $X$ is the set of non-zero non-negative $d$-dimensional real vectors, and a vector of $X$ is \emph{terminal} if it has only one non-zero entry:
$$
X=\mathbb{R}_+^d\backslash\{\vec{0}\}
\quad\textrm{and}\quad
X_0=\{\lambda \vec{e}_i~|~\lambda>0,~1\leq i\leq d\}.
$$
For $\vec{\alpha}\in X_0$, $B(\vec{\alpha})$ is the $d\times d$ identity matrix $I_d$, and for $\vec{\alpha}\in X\backslash X_0$, we set:
$$
B(\vec{\alpha})=
\left(\begin{array}{ccccc}
I_{i-1} & & & &\\
 & a & & 1 &\\
 & & I_{j-i-1} & &\\
 & 1 & & 0 &\\
 & & & & I_{d-j}
\end{array}\right),
$$
where $i$ and $j$ are the indexes of, respectively, the first and second largest entries of $\vec{\alpha}$ (with the smallest index being taken into account if several entries are equal), and $a$ is the integer part of the quotient of the $i$-th entry by the $j$-th one\footnote{Note that $\vec{\alpha}\notin X_0$ ensures that this $j$-th entry is non-zero, hence $a$ is well defined.}.
In other words, for $\vec{\alpha}=(\alpha_1,\ldots,\alpha_d)$, let $(i,j)$ be the lexicographically smallest pair of integers in $\{1,\ldots,d\}$ satisfying:
$$
\alpha_i=\max_{1\leq k\leq d}\alpha_k,
\qquad
\alpha_j=\max_{k\neq i}\alpha_k,
\qquad
a=\left\lfloor \frac{\alpha_i}{\alpha_j}\right\rfloor.
$$ Note that $a\geq 1$.
The map $T$ of Def.\ \ref{def:continued_fraction} is called the \emph{Brun map}. One has
for all $ \vec{\alpha} \in X$: 
$$ T(\vec{\alpha})= B(\vec{\alpha})^{-1} \, \vec{\alpha}.$$

According to this definition, each step of the Brun algorithm consists in, first, subtracting $a$ times the second largest entry of a vector of $X\backslash X_0$ to its first largest entry (that is, as many times as this first entry remains non-negative), and second, in exchanging both entries.
In the $d=2$ case, we thus recover the Euclidean algorithm.
Let us conclude by recalling that the Brun algorithm is weakly convergent (see, \emph{e.g.}, \cite{brentjes}).

\subsection{Brun expansions of stepped planes}
\label{sec:brun_exp_stepped_planes}

In this section, we show how Brun expansions of vectors can be naturally extended to Brun expansions of stepped planes, providing a new geometric viewpoint.
More generally, we want to define over binary functions a map $\tilde{T}$ whose restriction over the set of stepped planes satisfies, for any $\vec{\alpha}\in\mathbb{R}_+^d$ and $\rho\in\mathbb{R}$:
$$
\tilde{T}(\mathcal{P}_{\vec{\alpha},\rho})=\mathcal{P}_{T(\vec{\alpha}),\rho}.
$$
This can be easily done thanks to dual maps.
Indeed, for any invertible morphism $\beta(\vec{\alpha})$ whose incidence matrix is equal to the unimodular matrix $B(\vec{\alpha})$, Th.\ \ref{th:image_stepped_plane} yields, by using the fact that $B(\vec{\alpha})^\top=B(\vec{\alpha})$:
\begin{equation}\label{eq:actionbrun}
E_1^*(\beta(\vec{\alpha})^{-1})(\mathcal{P}_{\vec{\alpha},\rho})
=\mathcal{P}_{(B(\vec{\alpha})^{-1})^\top\vec{\alpha},\rho}
=\mathcal{P}_{B(\vec{\alpha}^{-1}\vec{\alpha},\rho}
=\mathcal{P}_{T(\vec{\alpha}),\rho}.
\end{equation}
For example, we can choose for $\beta(\vec{\alpha})$ the Brun substitution $\beta_{a,i,j}$ introduced in Ex.\ \ref{ex:brun_substitutions} (hence the name).
This is the choice we do throughout the present paper.\\

However, we also would like an \emph{effective} and \emph{geometric} definition of $\tilde{T}$ based on ``local'' data: the way the image of a stepped plane $\mathcal{P}$ is defined should rely only on $\mathcal{P}$ (seen as a binary stepped function) and not, as above, on its normal vector (which is not assumed to be known). 
Let us show that this can be done by introducing the notion of \emph{run}, which will allow us to give a multi-dimensional analogue of the fact that factors $00$ and $11$ cannot occur simultaneously in a Sturmian word:

\begin{definition}\label{def:run}
Let $i,j$ in $\{1,\ldots,d\}$ with $i\neq j$.
Let $\mathcal{B}$ be a binary stepped function.
An \emph{$(\vec{e}_j,i)$-run} of $\mathcal{B}$ is a non-zero binary function $\mathcal{R}$ which can be written:
$$
\mathcal{R}=\sum_{k\in I}(\vec{x}+k \vec{e}_j ,i^*),
$$
where $\vec{x}\in\mathbb{Z}^d$ and $I$ is maximal among the intervals of $\mathbb{Z}$ such that $\mathcal{R}\leq\mathcal{B}$.
Such a run is said to be of \emph{type} $(\vec{e}_j,i)$ and to have \emph{length} $\#I$.
\end{definition}

In other words, the geometric interpretation of a length $k$ $(\vec{e}_j,i)$-run of a binary function $\mathcal{B}$ is a maximal sequence of $k$ contiguous faces of type $i$, aligned with the direction $\vec{e}_j$ and included in the geometric interpretation of $\mathcal{B}$ (Fig.\ \ref{fig:runs}).\\

\begin{figure}[hbtp]
\centering
\includegraphics[width=\textwidth]{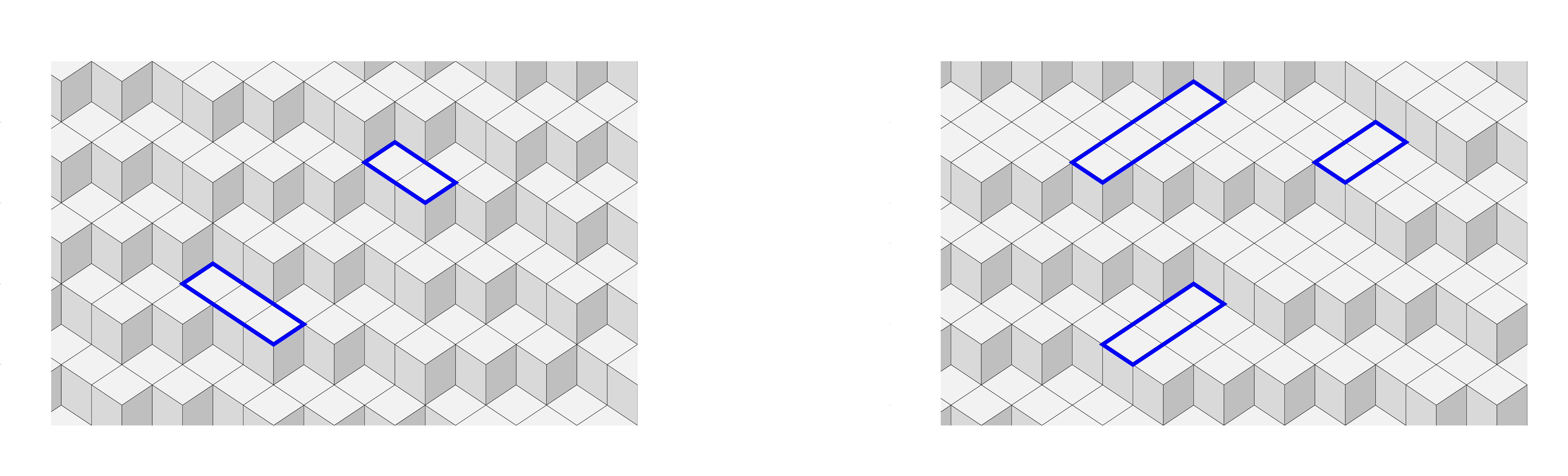}
\caption{Geometric interpretation of a stepped plane with $(\vec{e}_2,3)$-runs of length $2$ and $3$ (left, framed runs) and of a stepped surface with $(\vec{e}_1,3)$-runs of length $2$--$4$ (right, framed runs) -- see also Fig. \ref{fig:faces} for the directions $\vec{e}_i$
 and for the labeling of the faces.
}
\label{fig:runs}
\end{figure}

A stepped plane has $(\vec{e}_j,i)$-runs if and only if its geometric interpretation contains the geometric interpretation of a type $i$ face.
It follows from Def.\ \ref{def:stepped_plane} that this holds if and only if the $i$-th entry of the normal vector of this stepped plane is non-zero.
In this case, we can characterize the length of $(\vec{e}_j,i)$-runs:
\begin{proposition}\label{prop:runs_stepped_plane}
Let $i,j$ in $\{1,\ldots,d\}$ with $i\neq j$.
Let $\mathcal{P}_{\vec{\alpha},\rho}$ be a stepped plane, with $\vec{\alpha}=(\alpha_1,\ldots,\alpha_d)$ and $\alpha_i\neq 0$.
Consider the $(\vec{e}_j,i)$-runs of $\mathcal{P}_{\vec{\alpha},\rho}$.
If $\alpha_j=0$, these runs are all infinite.
If $\frac{\alpha_i}{\alpha_j}\in\mathbb{Z}$, they have all length $\max(\frac{\alpha_i}{\alpha_j},1)$.
Otherwise, there are runs of length $\max(\lfloor\frac{\alpha_i}{\alpha_j}\rfloor,1)$ and runs of length
$\max(\lfloor\frac{\alpha_i}{\alpha_j}\rfloor+1,1)$.
\end{proposition}

\begin{proof}
Since $\alpha_i\neq 0$, then $\mathcal{P}_{\vec{\alpha},\rho}$ has $(\vec{e}_j,i)$-runs.
Let $\mathcal{R}$ be an $(\vec{e}_j,i)$-run of $\mathcal{P}_{\vec{\alpha},\rho}$:
$$
\mathcal{R}=\sum_{k\in I}(\vec{x}+k \vec{e}_j ,i^*),
$$
where $\vec{x}\in\mathbb{Z}^d$ and $I\subset\mathbb{Z}$.
We first assume that $\alpha_j=0$. Then, according to Def.\ \ref{def:stepped_plane}, for any $k\in\mathbb{Z}$, $\langle\vec{x}+k \vec{e}_j|\vec{\alpha}\rangle=\langle\vec{x}|\vec{\alpha}\rangle$, \emph{i.e.}, $\mathcal{P}_{\vec{\alpha},\rho}(\vec{x},i)=\mathcal{P}_{\vec{\alpha},\rho}(\vec{x}+k\vec{e}_j,i)$.
By maximality of $I$, this implies $I=\mathbb{Z}$.\\
Let us now also assume $\alpha_j\neq 0$: $I$ is thus finite, say $I=\{b,\ldots,c\}$.
On the one hand, one checks:
$$
\mathcal{P}_{\vec{\alpha},\rho}(\vec{x}+b\vec{e}_j,i)=1 ~\Rightarrow~ \rho-\langle\vec{x}|\vec{\alpha}\rangle\in (b\alpha_j,b\alpha_j+\alpha_i],
$$
$$
\mathcal{P}_{\vec{\alpha},\rho}(\vec{x}+c\vec{e}_j,i)=1 ~\Rightarrow~ \rho-\langle\vec{x}|\vec{\alpha}\rangle\in (c\alpha_j,c\alpha_j+\alpha_i].
$$
The intervals $(b\alpha_j,b\alpha_j+\alpha_i]$ and $(c\alpha_j,c\alpha_j+\alpha_i]$ thus have
a non-empty intersection.
This yields:
$$
b\alpha_j+\alpha_i>c\alpha_j.
$$
On the other hand, one checks:
$$
\mathcal{P}_{\vec{\alpha},\rho}(\vec{x}+(b-1)\vec{e}_j,i)=0 ~\Rightarrow~ \rho-\langle\vec{x}|\vec{\alpha}\rangle\notin ((b-1)\alpha_j,(b-1)\alpha_j+\alpha_i],
$$
$$
\mathcal{P}_{\vec{\alpha},\rho}(\vec{x}+(c+1)\vec{e}_j,i)=0 ~\Rightarrow~ \rho-\langle\vec{x}|\vec{\alpha}\rangle\notin ((c+1)\alpha_j,(c+1)\alpha_j+\alpha_i].
$$
Since one already knows that $\rho-\langle\vec{x}|\vec{\alpha}\rangle\in(b\alpha_j, b\alpha_j+\alpha_i]$, this yields:
$$
(b-1)\alpha_j+\alpha_i<(c+1)\alpha_j.
$$
The two above inequalities allow us to bound the length $c-b+1$ of $I$:
$$
\frac{\alpha_i}{\alpha_j}-1<c-b+1<\frac{\alpha_i}{\alpha_j}+1,
$$
that is:
$$
c-b+1\in\left\{\left\lfloor\frac{\alpha_i}{\alpha_j}\right\rfloor,\left\lceil\frac{\alpha_i}{\alpha_j}\right\rceil\right\}, 
$$
with $\lfloor x \rfloor$ and $\lceil x \rceil$ standing, respectively, for the floor and the ceiling of $x$.\\
This proves that the $(\vec{e}_j,i)$-runs can only have the two claimed lengths\footnote{Recall that, by definition, runs have length at least $1$.}.\\
It remains to show that these lengths indeed occur.
Remark that for $\alpha_i\leq\alpha_j$ or $\alpha_i/\alpha_j=a\in\mathbb{N}$, both lengths are in fact equal (resp. to $1$ and $a$): there is nothing to show.
Let us now assume that $\alpha_j<\alpha_i$ and $\alpha_i/\alpha_j\notin\mathbb{N}$.
Suppose that all the $(\vec{e}_j,i)$-runs of $\mathcal{P}_{\vec{\alpha},\rho}$ have the same length, say $a$, and let us obtain a contradiction.
Consider the $(\vec{e}_i,j)$-runs of $\mathcal{P}_{\vec{\alpha},\rho}$: according to what preceeds, they exist and can only have length $\max(\lfloor\alpha_j/\alpha_i\rfloor,1)$ and $\max(\lceil\alpha_j/\alpha_i\rceil,1)$, that is, length $1$.
Roughly speaking, this enforces $\mathcal{P}_{\vec{\alpha},\rho}$ to alternate $(\vec{e}_j,i)$-runs of length $a$ and $(\vec{e}_i,j)$-runs of length $1$.
More precisely, $\mathcal{P}_{\vec{\alpha},\rho}$ is larger than the following (infinite) binary stepped function:
$$
\sum_{l\in\mathbb{Z}}\left((\vec{x}+l(a\vec{e}_j-\vec{e}_i),j^*)+\sum_{k=0}^{a-1}(\vec{x}+l(a\vec{e}_j-\vec{e}_i)+k\vec{e}_j,i^*)\right).
$$
In particular, this yields $\mathcal{P}_{\vec{\alpha},\rho}(\vec{x}+l(a\vec{e}_j-\vec{e}_i),i)=1$ for any $l\in\mathbb{Z}$, and thus:
$$
l(a\alpha_j-\alpha_i)<\rho-\langle\vec{x}|\vec{\alpha}\rangle\leq l(a\alpha_j-\alpha_i)+\alpha_i.
$$
By dividing by $l\neq 0$ and by taking the limit when $l$ goes to infinity, we get $a\alpha_j-\alpha_i=0$, \emph{i.e.}, $\alpha_i/\alpha_j=a$.
This contradicts the assumption $\alpha_i/\alpha_j\notin\mathbb{N}$.
 We deduce that $(\vec{e}_i,j)$-runs runs of both lengths $\lfloor \alpha_i/\alpha_j \rfloor$ and $\lceil \alpha_i/\alpha_j \rceil$ occur in $\mathcal{P}_{\vec{\alpha},\rho}$.
\end{proof}

The previous proposition allows us to deduce partial information about the normal vector of a stepped plane from its runs.
For example, one can compare the entries of the (unknown) normal vector $\vec{\alpha}=(\alpha_1,\ldots,\alpha_d)$ of a given stepped plane.
Indeed, one has $\alpha_i>\alpha_j$ if and only if $\lceil\alpha_i/\alpha_j\rceil\geq 2$, \emph{i.e.}, if and only if this stepped plane has an $(\vec{e}_j,i)$-run of length at least $2$.
More precisely, let us show that runs contain enough information to give an effective definition of the map $\tilde{T}$ (discussed in the beginning of this section).
For the sake of generality, we state a definition for binary stepped functions (as for runs -- recall Def.\ \ref{def:run}):

\begin{definition}\label{def:tilde_T}
Let $\tilde{T}$ be the map defined over binary stepped functions as follows.
Let $\mathcal{B}\in\mathfrak{B}$.
Assume that there are $i,j \in\{1,\ldots, d\}$ with $i\neq j$ satisfying the two conditions:
\begin{enumerate}
\item for any $k\neq i$, all $(\vec{e}_i,k)$- and $(\vec{e}_j,k)$-runs (if any) of $\mathcal{B}$ have length $1$;
\item $\mathcal{B}$ has at least one finite $(\vec{e}_j,i)$-run.
\end{enumerate}
Take for $(i,j)$ the lexicographically smallest pair which satisfies Cond. 1 and 2 (with $i \neq j)$, and let $a$ be the length of the shortest $(\vec{e}_j,i)$-run of $\mathcal{B}$.
We set:
$$
\tilde{T}(\mathcal{B})=E_1^*(\beta_{a,i,j}^{-1})(\mathcal{B}).
$$
Otherwise, we set $\tilde{T}(\mathcal{B})=\mathcal{B}$.
\end{definition}

\noindent This defines a map $\tilde{T}$ from $\mathfrak{B}$ to $\mathfrak{F}$ which satisfies the wanted property:

\begin{theorem}\label{th:tilde_T_stepped_plane}
For any $\vec{\alpha}\in\mathbb{R}_+^d\backslash\{\vec{0}\}$ and $\rho\in\mathbb{R}$, one has:
$$
\tilde{T}(\mathcal{P}_{\vec{\alpha},\rho})=\mathcal{P}_{T(\vec{\alpha}),\rho}.
$$
\end{theorem}

\begin{proof}
Let us consider a stepped plane $\mathcal{P}_{\vec{\alpha},\rho}$, and let us write $\vec{\alpha}=(\alpha_1,\ldots,\alpha_d)$.
First, let us assume that there are $i$ and $j$ satisfying the two conditions stated in Def.\ \ref{def:tilde_T}.
On the one hand, the first condition yields\footnote{with $\alpha_i\geq\alpha_j$ following from the $k=j$ case.}:
$$
\forall k,~\alpha_i\geq\alpha_j\geq\alpha_k.
$$
On the other hand, the second condition ensures $\alpha_j\neq 0$.
We deduce from Prop.\ \ref{prop:runs_stepped_plane} that $a=\lfloor \alpha_i/\alpha_j\rfloor$.
The result then follows from Th.\ \ref{th:image_stepped_plane} and from the definition of the Brun map $T$.\\
Let us now assume that there is no $i$ and $j$ satisfying the two conditions.
Let $\alpha_i$ and $\alpha_j$ be respectively the first and the second largest entry of $\vec{\alpha}$.
As previously, it follows from Prop.\ \ref{prop:runs_stepped_plane} that $i$ and $j$ satisfy the first condition.
They thus cannot satisfy the second one, that is, there is no finite $(\vec{e}_j,i)$-run.
Prop.\ \ref{prop:runs_stepped_plane} thus yields $\alpha_j=0$, that is, $\vec{\alpha}$ has only one non-zero entry.
In such a case, the definition of $T$ yields $T(\vec{\alpha})=\vec{\alpha}$.
The result follows.
\end{proof}


\begin{example}\label{ex:runs_stepped_plane}
Consider the stepped plane corresponding to Fig.\ \ref{fig:runs} (left).
One has:
\begin{center}
\begin{tabular}{|c||c|c|c|c|c|c|}
\hline
runs & $(\vec{e}_1,2)$ & $(\vec{e}_2,1)$ & $(\vec{e}_1,3)$ & $(\vec{e}_3,1)$ & $(\vec{e}_2,3)$ & $(\vec{e}_3,2)$\\ 
\hline
length & $1$ & $1,2$ & $1$ & $1,2$ & $2,3$ & $1$\\
\hline
\end{tabular}
\end{center}
The lexicographically smallest pair $(i,j)$ satisfying the two conditions stated in Def.\ \ref{def:tilde_T} is $(1,3)$.
Then, since the smallest $(\vec{e}_3,1)$-run has length $1$, $\tilde{T}$ will act on this stepped plane as the dual map $E_1^*(\beta_{1,1,3}^{-1})$.
\end{example}

\noindent Fig.\ \ref{fig:stepped_plane_exp} depicts some steps of the Brun expansion of a stepped plane with normal vector $\vec{\alpha}=(1,\frac{1}{19},\frac{25}{76})$.\\

\begin{figure}[hbtp]
\centering
\includegraphics[width=\textwidth]{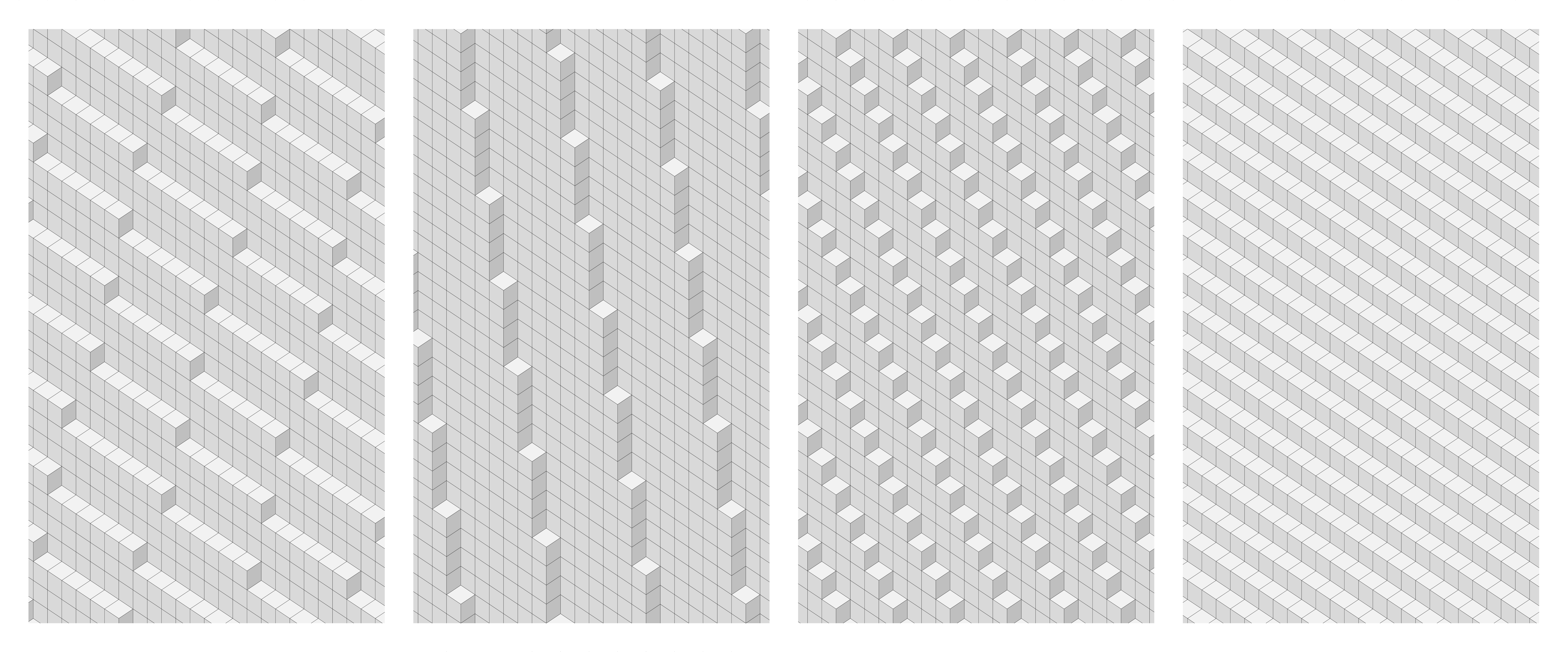}
\caption{
Iterated applications of the map $\tilde{T}$ on a stepped plane (from left to right).
}
\label{fig:stepped_plane_exp}
\end{figure}

\noindent Let us end by a simple remark:

\begin{remark}
Computing expansions of stepped planes as described in this section does not hold only for the Brun algorithm, but for any multi-dimensional continued fraction algorithm such that the matrix $B(\vec{\alpha})$ can be obtained by performing comparisons of entries of $\vec{\alpha}$ and by computing integer parts of quotients of these entries.
This turns out to concern most of known algorithms, as the Jacobi-Perron, Selmer or Poincar\'e ones.
\end{remark}

\subsection{Brun expansions of stepped surfaces}
\label{sec:brun_exp_stepped_surfaces}

In the previous sections, an effective way to compute the Brun expansion of a stepped plane has been provided, thanks to the map $\tilde{T}$ (Def.\ \ref{def:tilde_T}, Th.\ \ref{th:tilde_T_stepped_plane}).
Here, we show that the fact that $\tilde{T}$ has not been defined only over stepped planes but over binary stepped functions -- among them stepped surfaces -- allows us to extend the notion of Brun expansion to stepped surfaces.
The keypoint is the following result:

\begin{theorem}\label{th:tilde_T_stepped_surface}
The image under $\tilde{T}$ of any stepped surface is a stepped surface.
\end{theorem}

\begin{proof}
Let $\mathcal{S}$ be a stepped surface.
If $\tilde{T}(\mathcal{S})=\mathcal{S}$, then there is nothing to prove.
Otherwise, let $i$ and $j$ be distinct integers satisfying the two conditions stated in Def.\ \ref{def:tilde_T}.
Let then $a$ be the length of the shortest $(\vec{e}_j,i)$-run of $\mathcal{S}$.
Thus, by definition, $\tilde{T}(\mathcal{S})=E_1^*(\beta_{a,i,j}^{-1})(\mathcal{S})$.\\
Let us prove that, for all $\vec{x}\in \mathbb{Z}^d$ such that $(\vec{x},j^*)$ has weight one in $\mathcal{S}$, $(\vec{x},i^*)$ has also weight one in $\mathcal{S}$.
Assume that $(\vec{x},j^*)$ has weight one in $\mathcal{S}$.
Consider the $(d-2)$-dimensional facet of the geometric interpretation of $(\vec{x},j^*)$ that contains the vectors $\vec{x}+ \vec{e}_i+\vec{e_j}+ \varepsilon_k \vec{e}_k$, with $\varepsilon_k \in \{0,1\}$, $k \in \{1,\ldots,d\}$ and $k \notin\{i,j\}$.
This $(d-2)$-dimensional facet is contained in some other facet of unit hypercube which is the geometric interpretation of some face $(\vec{y},\ell^*)$ with weight one in $\mathcal {S}$.
Necessarily, one has $\ell=i$ or $\ell=j$.
If $\ell=j$, then $\vec{y}=\vec{x}+\vec{e}_i$.
But this contradicts the fact that $(\vec{e}_i,j)$-runs have length $1$.
If $\ell=i$, then either $\vec{y}=\vec{x}$ or $\vec{y}=\vec{x}+\vec{e}_j$.
The latter case is impossible because the intersection of the projections of the geometric interpretation of faces $(\vec{x}, j^*)$ and $(\vec{x}+\vec{e}_j, i^*)$ would have a non-empty interior, which would contradict Def.\ \ref{def:stepped_surface}.
Thus, one has proven that $\vec{y}=\vec{x}$, \emph{i.e.}, each face of $\mathcal{S}$ of type $j$ located at some $\vec{x}\in\mathbb{Z}^d$ is ``followed'' by the face $(\vec{x}, i^*)$, and thus by the faces $(\vec{x}-k\vec{e}_j,i^*)$, for $0\leq k<a$.\\
We can now write $\mathcal{S}$ as a disjoint sum of binary stepped functions $\mathcal{C}_{\vec{x}}=(\vec{x},j^*)+\sum_{\ell=0}^{a-1}(\vec{x}-\ell\vec{e}_j,i^*)$, of faces of type $k\notin\{i,j\}$, and of faces of type $i$ (those occuring in the $(\vec{e}_j,i)$-runs of length larger than or equal to $a$).
This is illustrated in Fig.\ \ref{fig:decompose}.
Then, one computes (recall Ex.\ \ref{ex:brun_substitutions}):
$$
E_1^*(\beta_{a,i,j}^{-1})~:~\left\{\begin{array}{cll}
\mathcal{C}_{\vec{x}} & \mapsto & (B(a\vec{e}_i+\vec{e}_j)\vec{x},i^*),\\
(\vec{x},k^*) & \mapsto & (B(a\vec{e}_i+\vec{e}_j)\vec{x},k^*),\\
(\vec{x},i^*) & \mapsto & (B(a\vec{e}_i+\vec{e}_j)\vec{x}+a\vec{e}_i,j^*).
\end{array}\right.
$$
Since $B(a\vec{e}_i+\vec{e}_j)$ is invertible, any two of these image faces have different locations whenever they have the same type.
This yields that $\tilde{T}(\mathcal{S})=E_1^*(\beta_{a,i,j}^{-1})(\mathcal{S})$ is a binary stepped function, and thus a stepped surface according to Th.\ \ref{th:image_stepped_surface}.
\end{proof}

\begin{figure}[hbtp]
\centering
\includegraphics[width=\textwidth]{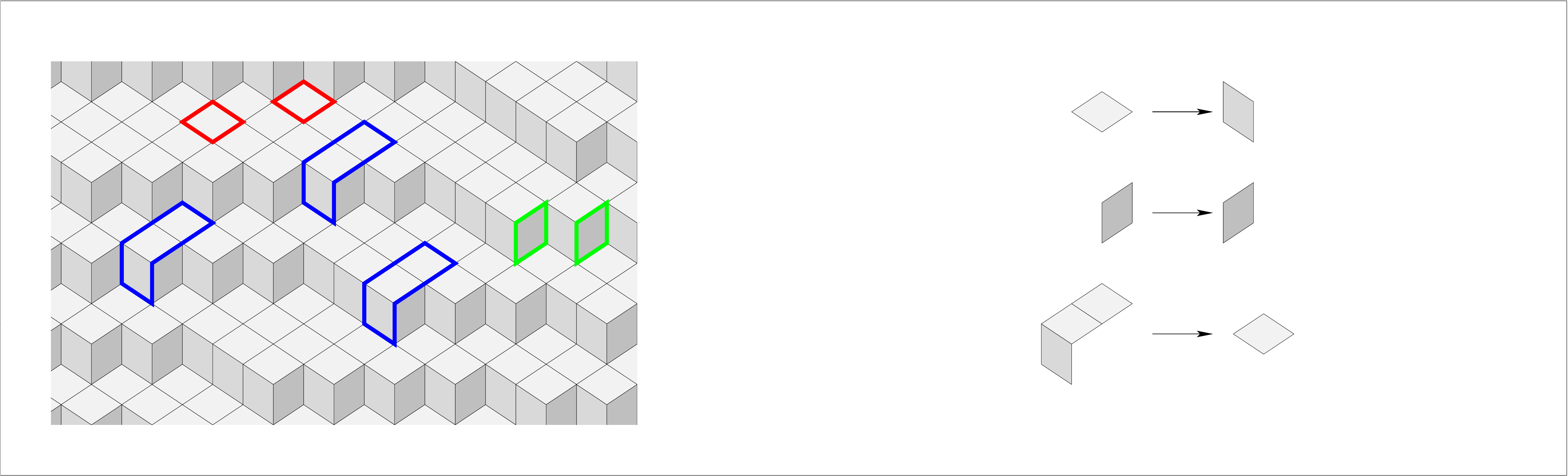}
\caption{
A stepped surface which has $(\vec{e}_j,i)$-runs of length larger than or equal to $a$ can be written as a disjoint sum of binary stepped functions $\mathcal{C}_{\vec{x}}=(\vec{x},j^*)+\sum_{\ell=0}^{a-1}(\vec{x}-\ell\vec{e}_j,i^*)$, of faces of type $k\notin\{i,j\}$, and of faces of type $i$ (left, with $(i,j)=(3,1)$ and $a=2$).
The image of these three particular binary stepped functions (right) turning out to be disjoint, the image of the whole stepped surface is a binary stepped function.
}
\label{fig:decompose}
\end{figure}

\begin{example}\label{ex:runs_stepped_surface}
Consider the stepped surface of Fig.\ \ref{fig:runs} (right) or \ref{fig:decompose}.
One has:
\begin{center}
\begin{tabular}{|c||c|c|c|c|c|c|}
\hline
runs & $(\vec{e}_1,2)$ & $(\vec{e}_2,1)$ & $(\vec{e}_1,3)$ & $(\vec{e}_3,1)$ & $(\vec{e}_2,3)$ & $(\vec{e}_3,2)$\\ 
\hline
length & $1$ & $1,2,3,5$ & $2,3,4$ & $1$ & $2,3,4,6$ & $1$\\
\hline
\end{tabular}
\end{center}
The lexicographically smallest pair $(i,j)$ satisfying the two conditions stated in Def.\ \ref{def:tilde_T} is $(3,1)$.
Then, since the smallest $(\vec{e}_1,3)$-run has length $2$, $\tilde{T}$ will act on this stepped surface as the dual map $E_1^*(\beta_{2,1,3}^{-1})$.
According to Th.\ \ref{th:tilde_T_stepped_surface}, the resulting image is a stepped surface.
\end{example}

\noindent We can then naturally define the \emph{Brun expansion} of a stepped surface:

\begin{definition}\label{def:stepped_surface_exp}
The \emph{Brun expansion} of a stepped surface $\mathcal{S}$ is the sequence of matrices $(B_n)_{n\geq 1}$ defined as follows.
If $\tilde{T}^{n+1}(\mathcal{S})=\tilde{T}^n(\mathcal{S})$, then $B_n$ is the $d\times d$ identity matrix $I_d$.
Otherwise, $B_n$ is the incidence matrix of the Brun substitution $\beta_n$ such that $\tilde{T}^{n+1}(\mathcal{S})=E_1^*(\beta_n^{-1})(\tilde{T}^n(\mathcal{S}))$.
The \emph{length} of this Brun expansion is the smallest $n$ such that $B_n=I_d$.
\end{definition}

In particular, if $\mathcal{S}$ is a stepped plane, then we recover the notion of Brun expansion introduced in the previous section.
Fig.\ \ref{fig:stepped_surface_exp} depicts some iterations of the map $\tilde{T}$ on a stepped surface. These stepped surfaces are not stepepd planes. Indeed, the second iterate (i.e., the third picture starting from the left)
has $(\vec{e}_2,2^*)$ runs of three different lengths.

\begin{figure}[hbtp]
\centering
\includegraphics[width=\textwidth]{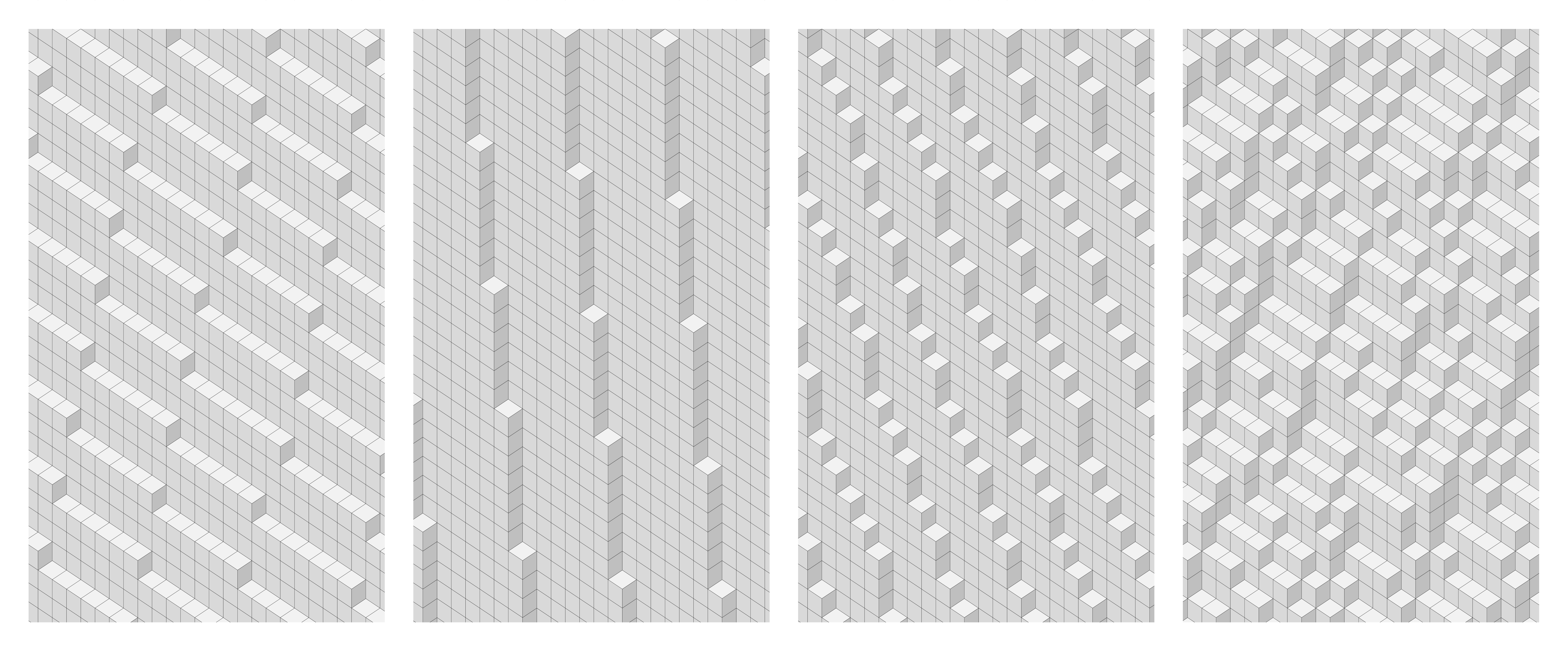}
\caption{Iterated applications of the map $\tilde{T}$ on a stepped surface (from left to right).
Note that stepped surfaces are more and more bumby, due to the contracting effect of $\tilde{T}$.
}
\label{fig:stepped_surface_exp}
\end{figure}

\begin{remark}
Since $\tilde{T}$ has been defined over binary stepped functions, it seems more natural to extend the notion of Brun expansion not to stepped surfaces but to binary stepped functions.
However, one checks that the image under $\tilde{T}$ of a binary stepped function is not always binary, \emph{i.e.}, Th.\ \ref{th:tilde_T_stepped_surface} cannot be extended.
This prevents us from defining Brun expansions of general binary stepped functions.
\end{remark}

\subsection{Common expansions}
\label{sec:common}

According to Th.\ \ref{th:tilde_T_stepped_plane}, a real vector and a stepped plane share the same Brun expansion if and only if the former is the normal vector of the latter.
However, there is no such obvious link for stepped surfaces.
We thus address the following question: what is the set $\mathfrak{S}_{\vec{\alpha}}$ of the stepped surfaces having the same Brun expansion as $\vec{\alpha}\in\mathbb{R}_+^d\backslash\{\vec{0}\}$?\\

\noindent We first introduce \emph{stepped quasi-planes}, that will play a specific role with respect to the previous question:

\begin{definition}\label{def:stepped_quasi_plane}
Let $\vec{\alpha}\in\mathbb{R}_+^d\backslash\{\vec{0}\}$ and $\rho\in\mathbb{R}$.
Let $D\subset\{\vec{x}\in\mathbb{Z}^d~|~\langle\vec{x}|\vec{\alpha}\rangle=\rho\}$.
If the following stepped function $ \mathcal{Q}$ is a stepped surface, then it is said to be a \emph{stepped quasi-plane} of \emph{normal vector} $\vec{\alpha}$, \emph{intercept} $\rho$ and \emph{defect} $D$:
$$
\mathcal{Q}=\mathcal{P}_{\vec{\alpha},\rho}+\sum_{\vec{x}\in D}\mathcal{F}_{\vec{x}}.
$$
\end{definition}

The difference between a stepped plane and a stepped quasi-plane which have the same normal vector and intercept thus consists of some flips lying in a set containing some of the vertices of the geometric realisation of the stepped plane (the ``defect''), as illustrated in Fig.\ \ref{fig:stepped_quasi_plane}.\\

\begin{figure}[hbtp]
\centering
\includegraphics[width=\textwidth]{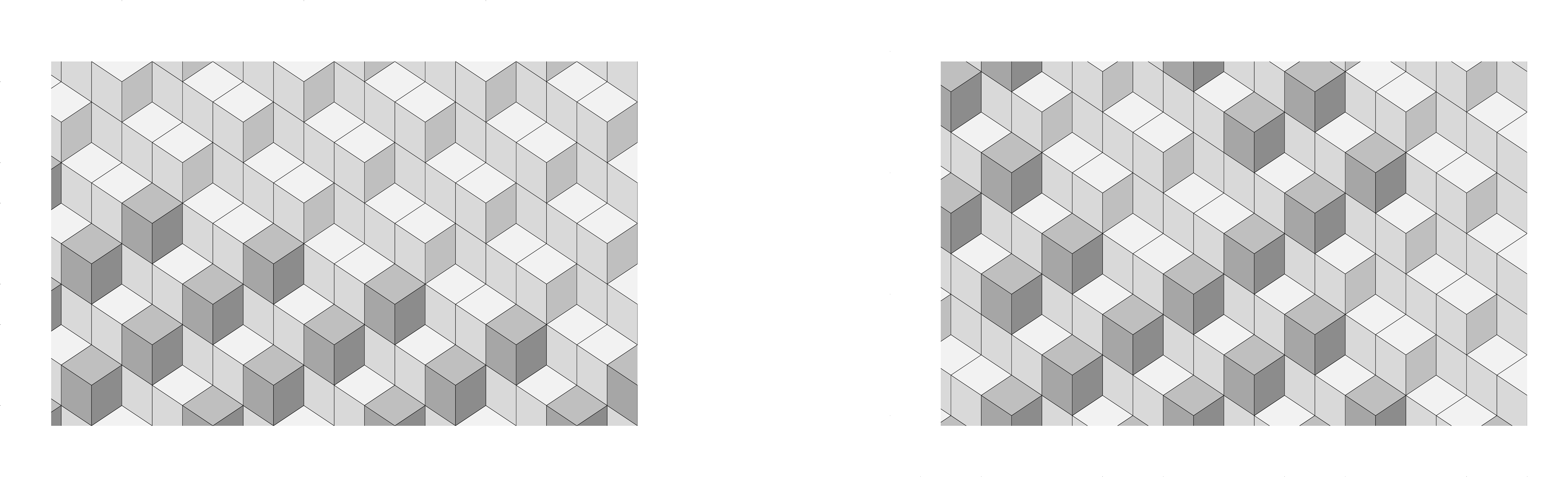}
\caption{Two stepped quasi-planes of normal vector $(3,1,2)$, with different defects (corresponding to the shaded upper facets of unit cubes).}
\label{fig:stepped_quasi_plane}
\end{figure}

Note that, in Def.\ \ref{def:stepped_quasi_plane}, not any defect $D$ suits: in some cases, one gets non-binary stepped functions.
Consider, for example, the case of a single flip added to a stepped plane of normal vector $\vec{e}_i$: since this stepped plane has only type $i$ faces, the added flip yields negative-weighted faces.
However, there is a rather simple characterization of suitable defects:

\begin{proposition}\label{prop:suitable_defect}
Let $\vec{\alpha}$, $\rho$, $D$ and $\mathcal{Q}$ be as in Def.\ \ref{def:stepped_quasi_plane}.
Then, $\mathcal{Q}$ is a stepped surface (and thus a stepped quasi-plane) if and only if one has, for any $i$:
$$
\langle \vec{\alpha}|\vec{e}_i\rangle=0 ~\Rightarrow~ D\subset D+\vec{e}_i.
$$
\end{proposition}

\begin{proof}
For $(\vec{x},i)\in\mathbb{Z}^d\times\{1,\ldots,d\}$, one computes:
\begin{center}
\begin{tabular}{|c|c|c||c|}
\hline
$\mathcal{P}_{\vec{\alpha},\rho}(\vec{x},i)$ & $\vec{x}\in D$ ? & $\vec{x}+\vec{e}_i\in D$ ?& $\mathcal{Q}(\vec{x},i)$\\
\hline
0 & no & no & 0\\
0 & no & yes & -1\\
0 & yes & no & 1\\
0 & yes & yes & 0\\
1 & no & no & 1\\
1 & no & yes & 0\\
1 & yes & no & 2\\
1 & yes & yes & 1\\
\hline
\end{tabular}
\end{center}
Assume that $\langle \vec{\alpha}|\vec{e}_i\rangle=0$ yields $D\subset D+\vec{e}_i$.
The above table shows that $\mathcal{Q}$ takes values in $\{0,1\}$, except in the second and seventh cases.
In the second case, $\vec{x}+\vec{e}_i\in D$ yields $\langle\vec{x}+\vec{e}_i|\vec{\alpha}\rangle=\rho$, and $\mathcal{P}_{\vec{\alpha},\rho}(\vec{x},i)=0$ yields 
either $\langle\vec{x}|\vec{\alpha}\rangle\geq\rho$, or $\langle\vec{x}+\vec{e}_i|\vec{\alpha}\rangle <\rho$.
This ensures $\langle \vec{\alpha}|\vec{e}_i\rangle=0$, and thus $D\subset D+\vec{e}_i$.
But it is not compatible with $\vec{x}\notin D$ and $\vec{x}+\vec{e}_i \in D$: this case is ruled out.
In the seventh case, $\mathcal{P}_{\vec{\alpha},\rho}(\vec{x},i)=1$ yields $\langle\vec{x},\vec{\alpha}\rangle<\rho$.
But $\vec{x}\in D$ yields $\langle\vec{x}|\vec{\alpha}\rangle=\rho$: this case is also ruled out.
This proves that $\mathcal{Q}$ is a binary function.
Prop.\ \ref{prop:binary_pseudo_flip_acc} then ensures that it is a stepped surface, as it is pseudo-flip-accessible from a stepped plane.
The ``if'' part is proven.\\
Conversely, assume that $\mathcal{Q}$ is a stepped quasi-plane.
It is thus a binary function.
This rules out the second and the seventh cases.
One then checks that all the other cases are compatible with $D\subset D+\vec{e}_i$, except the sixth one.
In this case, $\mathcal{P}_{\vec{\alpha},\rho}(\vec{x},i)=1$ yields $\langle\vec{x}|\vec{\alpha}\rangle<\rho$, and $\vec{x}+\vec{e}_i\in D$ yields $\langle\vec{x}+\vec{e}_i|\vec{\alpha}\rangle=\rho$.
This ensures that $\langle \vec{\alpha}|\vec{e}_i\rangle$ is non-zero.
The ``only if'' part is proven.
\end{proof}

We can now state the main result of this section (recall that $\mathfrak{S}_{\vec{\alpha}}$ stands for the set of stepped surfaces having the same Brun expansion as $\vec{\alpha}\in\mathbb{R}_+^d$):

\begin{theorem}\label{th:weak_cv_surfaces}
Let $\vec{\alpha}\in\mathbb{R}_+^d\backslash\{\vec{0}\}$.
For any $\varepsilon>0$, there is $N\in \mathbb{N}$ such that if $\vec{\alpha}$ has a Brun expansion of length at least $N$, then any stepped surface in $\mathfrak{S}_{\vec{\alpha}}$ is at distance at most $\varepsilon$ from a stepped quasi-plane of normal vector $\vec{\alpha}$.
\end{theorem}

In other words, the longer the Brun expansion of a stepped surface is, the closer it is to a stepped quasi-plane whose normal vector has the same Brun expansion.
Before proving this theorem, let us state a corollary in the particular case of infinite Brun expansions:

\begin{corollary}\label{cor:tot}
If a stepped surface has the same infinite Brun expansion as a vector $\vec{\alpha}\in\mathbb{R}_+^d\backslash\{\vec{0}\}$, then it is a stepped quasi-plane of normal vector $\vec{\alpha}$.
If, moreover, $\vec{\alpha}$ is totally irrational, \emph{i.e.}, $\dim_{\mathbb{Q}}(\alpha_1,\ldots,\alpha_d)=d$, then this stepped surface is a ``generalized'' stepped plane, \emph{i.e.}, it matches Def.\ \ref{def:stepped_plane} but with strict and non-strict inequalities possibly interchanged.
\end{corollary}

\begin{proof}
The first statement is a direct consequence of Th.\ \ref{th:weak_cv_surfaces}.
Assume now that $\vec{\alpha}$ is totally irrational.
Then, $\langle\vec{y}|\vec{\alpha}\rangle=\rho$ is satisfied for at most one $\vec{y}\in\mathbb{Z}^d$, hence the defect $D$ has cardinality at most one.
If $D$ is empty, then one gets the stepped plane $\mathcal{P}_{\vec{\alpha}, \rho}$ (recall Def.\ \ref{def:stepped_plane}).
Otherwise, one gets the stepped surface $\mathcal{P}_{\vec{\alpha}, \rho}+ \mathcal{F}_{\vec{y}}$, which is a ``generalized'' stepped plane since it satisfies:
$$
(\mathcal{P}_{\vec{\alpha}, \rho}+\mathcal{F}_{\vec{y}})(\vec{x},i)=1 ~\Leftrightarrow~ \langle\vec{x}|\vec{\alpha}\rangle \leq \rho <\langle\vec{x}+\vec{e}_i|\vec{\alpha}\rangle.
$$
Note that the set of vertices of the geometric interpretation of such a ``generalized'' stepped plane is still a digitization of a Euclidean plane, namely a standard arithmetic discrete plane.
\end{proof}
 
Let us now prove Th.\ \ref{th:weak_cv_surfaces}.
We first need the following technical lemma which, informally, translates weak convergence of the Brun algorithm for real vectors (Def.\ \ref{def:weak_cv}) into flip divergence for stepped surfaces:

\begin{lemma}\label{lem:transfer}
Let $\vec{\alpha}\in\mathbb{R}_+^d\backslash\{\vec{0}\}$.
For any $R\geq 0$, there is $N\in\mathbb{N}$ such that if the Brun expansion $(B_n)_{n\geq 1}$ of $\vec{\alpha}$ has length at least $N$, then any integer vector $\vec{x}$ which is neither positive nor negative satisfies:
$$
\langle B_1^{-1}\ldots B_N^{-1}\vec{x}|\vec{\alpha}\rangle\neq 0
~\Rightarrow~
||B_1^{-1}\ldots B_N^{-1}\vec{x}||\geq R.
$$
\end{lemma}

\begin{proof}
Let $\vec{x}=(x_1,\ldots,x_d)$ be an integer vector which is neither positive nor negative. It admits an orthogonal vector $\vec{y}\in\mathbb{R}_+^d\backslash\{\vec{0}\}$.
Indeed, if there is $i$ such that $x_i=0$, then $\vec{y}=\vec{e}_i$ suits.
Otherwise, there are $i$ and $j$ such that $x_i>0$ and $x_j<0$, and one checks that $\vec{y}=x_i\vec{e}_i-x_j\vec{e}_j$ suits.
Then, $B_n=B_n^\top$ yields:
\begin{eqnarray*}
  0=\langle\vec{x}|\vec{y}\rangle
  &=&\langle B_n^{-1}\vec{x}|B_n\vec{y}\rangle\\
  &=&\langle B_{n-1}^{-1}B_n^{-1}\vec{x}|B_{n-1} B_n\vec{y}\rangle\\
&=&\ldots\\
&=&\langle B_1 ^{-1} \ldots B_{n-1}^{-1}B_n^{-1}\vec{x}|B_1 \ldots B_{n-1}B_n\vec{y}\rangle
\end{eqnarray*}
Thus, for any $n$, the vectors $\vec{x}_n$ and $\vec{y}_n$ defined by $\vec{x}_n=B_1^{-1}\ldots B_n^{-1}\vec{x}$ and $\vec{y}_n=B_1\ldots B_n\vec{y}$ are orthogonal.
The point is that weak convergence of the Brun algorithm ensures that the sequence $(\vec{y}_n)_n$ tends in direction towards $\vec{\alpha}$.
What does it yield for the sequence $(\vec{x}_n)_n$?
To answer this question, let us introduce the set $V$ of neither positive nor negative integer vectors which do not belong to $\vec{\alpha}^\bot$.
For any $R>0$, the number $\varepsilon_R$ defined as follows is positive:
$$
\varepsilon_R=\min_{\vec{z}\in V\cap B(\vec{0},R)}d\left(\frac{\vec{z}}{||\vec{z}||},\vec{\alpha}^\bot\right).
$$
Moreover, one checks that any two orthogonal unitary vectors $\vec{u}$ and $\vec{v}$ satisfy:
$$
d(\vec{u},\vec{\alpha}^\bot)\leq d(\vec{v},\mathbb{R}\vec{\alpha}).
$$
Indeed, let us write $\vec{u}=\vec{u}_{\vec{\alpha}}+\vec{u}_{\vec{\alpha} ^\bot}$ with $\vec{u}_{\vec{\alpha}}\in\mathbb{R}\vec{\alpha}$ and $\vec{u}_{\vec{\alpha} ^ \bot}\in\vec{\alpha} ^\bot$ and, similarly, $\vec{v}=\vec{v}_{\vec{\alpha}}+\vec{v}_{\vec{\alpha} ^ \bot}$.
We deduce from $\langle\vec{u}|\vec{v}\rangle=0$ that
$$
||\vec{u}_{\vec{\alpha} ^\bot }|| \cdot ||\vec{v}_{\vec{\alpha} ^\bot}||
\geq
|\langle\vec{u}_{\vec{\alpha} ^ \bot}|\vec{v}_{\vec{\alpha} ^\bot}\rangle|
=
|\langle\vec{u}_{\vec{\alpha}} |\vec{v}_{\vec{\alpha}}\rangle|
=
||\vec{u}_{\vec{\alpha}} || \cdot ||\vec{v}_{\vec{\alpha}}||.
$$
Thus
$$
\frac{ ||\vec{v}_{\vec{\alpha} ^\bot } ||}{ ||\vec{v}_{\vec{\alpha}} ||} \geq \frac{ ||\vec{u}_{\vec{\alpha} } ||}{ ||\vec{u}_{\vec{\alpha} ^\bot} ||}.$$
This yields
$$
\frac{ ||\vec{v}_{\vec{\alpha} ^\bot } ||^2}{ 1- ||\vec{v}_{\vec{\alpha} ^\bot } ||^2} \geq \frac{ ||\vec{u}_{\vec{\alpha}} ||^2}{ 1- ||\vec{u}_{\vec{\alpha} } ||^2},
$$
and the claim follows, since the mapping $x\mapsto \frac{x}{1-x}$ is monotonically increasing:
$$
d(\vec{u},\vec{\alpha}^\bot)= ||\vec{u}_{\vec{\alpha} } || \leq ||\vec{v}_{\vec{\alpha}^\bot} ||=d(\vec{v},\mathbb{R}\vec{\alpha}).
$$
Let us now fix $R>0$.
According to the weak convergence of the Brun algorithm, there is $N\in\mathbb{N}$ such that $d(\vec{y}_N/||\vec{y}_N||,\mathbb{R}\vec{\alpha})<\varepsilon_R$.
The above inequality then ensures $d(\vec{x}_N/||\vec{x}_N||,\vec{\alpha}^\bot)<\varepsilon_R$, and the definition of $\varepsilon_R$ yields $\vec{x}_N\notin V\cap B(\vec{0},R)$.
Hence, $\vec{x}_N\in V$ implies $||\vec{x}_N||\geq R$.
This ends the proof.
\end{proof}

\noindent We also need the following lemma, which derives from Prop.\ \ref{prop:stepped_surface} and \ref{prop:pseudo_flip_acc}:

\begin{lemma}\label{lem:zarbi}
Let $\mathcal{S}$ and $\mathcal{S}'$ be two stepped surfaces whose geometric interpretations both contain $\vec{0}$.
Then, there exist a sequence $(\vec{x}_n)_n$ of neither positive nor negative integer vectors and a sequence $(\varepsilon_n)_n$ with values in $\{\pm 1\}$ such that:
$$
\mathcal{S}'=\lim_{n\to\infty}\mathcal{S}+\sum_{k\leq n}\varepsilon_k\mathcal{F}_{\vec{x}_k}.
$$
\end{lemma}

\begin{proof}
According to Prop.\ \ref{prop:pseudo_flip_acc}, there exist a sequence $(\vec{x}_n)_n$ of integer vectors and a sequence $(\varepsilon_n)_n$ with values in $\{\pm 1\}$ such that:
$$
\mathcal{S}'=\lim_{n\to\infty}\mathcal{S}+\sum_{k\leq n}\varepsilon_k\mathcal{F}_{\vec{x}_k}.
$$
Moreover, by looking more carefully at the proof of Prop.\ \ref{prop:pseudo_flip_acc}, one sees that any flip $\mathcal{F}_{\vec{x}_n}$ is located \emph{between} $\mathcal{S}$ and $\mathcal{S}'$.
More precisely, there are two integer vectors $\vec{y}_n$ and $\vec{z}_n$ in the geometric interpretations of $\mathcal{S}$ and $\mathcal{S}'$, respectively, such that $\vec{x}_n=\lambda_n\vec{y}_n+(1-\lambda_n)\vec{z}_n$, with either $\lambda_n\in [0,1)$ if $\vec{x}_n$ is above $\mathcal{S}$ and strictly below $\mathcal{S}'$, or $\lambda_n\in (0,1]$ if $\vec{x}_n$ is above $\mathcal{S}'$ and strictly below $\mathcal{S}$ (only these cases occur).
Then, $\vec{x}_n>0$ yields $\vec{z}_n>0$ in the former case and $\vec{y}_n>0$ in the latter case, while $\vec{x}_n<0$ yields $\vec{y}_n<0$ in the former case and $\vec{z}_n<0$ in the latter case.
Prop.\ \ref{prop:stepped_surface} together with the fact that the geometric interpretations of $\mathcal{S}$ and $\mathcal{S}'$ both contain $\vec{0}$ ensure that $\vec{y}_n$ and $\vec{z}_n$ are neither positive nor negative, which is thus also the case for $\vec{x}_n$.
The result follows.
\end{proof}

\noindent We are now in a position to prove Th.\ \ref{th:weak_cv_surfaces}:

\begin{proof}
Let us fix $\varepsilon>0$.
Let $N$ be a positive integer. We will determine $N$ more precisely later.
Let $\mathcal{S}\in\mathfrak{S}_{\vec{\alpha}}$ and let $(B_n)_{n\geq 1}$ denote the common Brun expansion of $\mathcal{S}$ and $\vec{\alpha}$.
On the one hand, one has
$$
\tilde{T}^N(\mathcal{S})
=E_1^*(\beta_N^{-1})\circ\ldots\circ E_1^*(\beta_1^{-1})(\mathcal{S})
=E_1^*((\beta_N\circ\ldots\circ\beta_1)^{-1})(\mathcal{S}),
$$
where $\beta_n$ is the Brun substitution of incidence matrix $B_n$.
Since, for two unimodular morphisms $\sigma$ and $\sigma'$, one has $E_1^*(\sigma\circ\sigma')=E_1^*(\sigma')\circ E_1^*(\sigma)$, this yields
$$
\mathcal{S}=E_1^*(\beta_N\circ\ldots\circ\beta_1) (\tilde{T}^N(\mathcal{S})).
$$
On the other hand, since $B_N^{-1}\ldots B_1^{-1}\vec{\alpha}=T^N(\vec{\alpha})\in\mathbb{R}_+^d\backslash\{\vec{0}\}$, Prop.\ \ref{prop:pseudo_flip_acc} ensures that the stepped surface $\tilde{T}^N(\mathcal{S})$ can be written
$$
\tilde{T}^N(\mathcal{S})=\mathcal{P}_{B_N ^{-1} \ldots B_1^{-1}\vec{\alpha},0}+\sum_{\vec{x}\in D_N}\varepsilon_{\vec{x}}\mathcal{F}_{\vec{x}},
$$
where $D_N\subset\mathbb{Z}^d$ and $\varepsilon_{\vec{x}}=\pm 1$ for all
$\vec{x} \in D_N$.
Note that $D_N$ has \emph{a priori} no reason for being a defect.
Moreover, since the Brun expansion of a stepped surface is translation-invariant, one can assume that $\vec{0}$ belongs to the geometric interpretation of $\tilde{T}^N(\mathcal{S})$.
Note that $\vec{0}$ also belongs to the geometrical interpretation of $\mathcal{P}_{B_N^{-1}\ldots B_1^{-1}\vec{\alpha},0}$, since any face $(\vec{0}-\vec{e}_i, i^*)$ belongs to a stepped plane of intercept $0$ (hence the choice of intercept $0$ here).
Thus, according to Lem.\ \ref{lem:zarbi}, one can assume that vectors of $D_N$ are neither positive nor negative.
Let us now go back to $\mathcal{S}$.
One deduces from the two above equalities:
$$
\mathcal{S}=E_1^*(\beta_N\circ\ldots\circ\beta_1)\left(\mathcal{P}_{B_N^{-1}\ldots B_1^{-1}\vec{\alpha},0}+\sum_{\vec{x}\in D_N}\varepsilon_{\vec{x}}\mathcal{F}_{\vec{x}}\right).
$$
Using Th.\ \ref{th:image_stepped_plane} and Prop.\ \ref{prop:image_flip} and that the matrices $B_n$ are symmetric, this yields
$$
\mathcal{S}=\mathcal{P}_{\vec{\alpha},0}+\sum_{\vec{x}\in D_N}\varepsilon_{\vec{x}}\mathcal{F}_{B_1^{-1}\ldots B_N^{-1}\vec{x}}.
$$
Let us now assume that $N$ is given by Lem.\ \ref{lem:transfer} with $R=-\log_2(\varepsilon/2)$.
We split $D_N$ into two parts $D_N'$ and $D_N\backslash D_N'$, where
$$
D_N'=\{\vec{x}\in D_N~|~\langle B_1^{-1}\ldots B_N^{-1}\vec{x}|\vec{\alpha}\rangle= 0\textrm{ and }||B_1^{-1}\ldots B_N^{-1}\vec{x}||\leq R\}.
$$
This leads to rewrite $\mathcal{S}$ as follows:
$$
\mathcal{S}=\mathcal{Q}_N+\sum_{\vec{x}\in D_N\backslash D'_N}\varepsilon_{\vec{x}}\mathcal{F}_{B_1^{-1}\ldots B_N^{-1}\vec{x}}
\quad\textrm{with}\quad
\mathcal{Q}_N=\mathcal{P}_{\vec{\alpha},0}+\sum_{\vec{x}\in D_N'}\varepsilon_{\vec{x}}\mathcal{F}_{B_1^{-1}\ldots B_N^{-1}\vec{x}}.
$$
On the one hand, since vectors of $D_N$ (hence of $D_N\backslash D_N'$) are neither positive nor negative, one can apply Lem.\ \ref{lem:transfer}, which ensures that the flips located at $B_1^{-1}\ldots B_N^{-1}\vec{x}$ for $\vec{x}\in D_N\backslash D_N'$ are located outside $B(\vec{0},R)$.
Thus
$$
d_{\mathfrak{F}}(\mathcal{S},\mathcal{Q}_N)\leq 2^{-R}=\varepsilon/2.
$$
On the other hand, for $\vec{x}\in D_N'$ and $\vec{y}=B_1^{-1}\ldots B_N^{-1}\vec{x}$, one has
$$
\langle\vec{y}|\vec{\alpha}\rangle=\langle B_1^{-1}\ldots B_N^{-1}\vec{x}|\vec{\alpha}\rangle=0
\quad\textrm{and}\quad
||\vec{y}||\leq R.
$$
We would like to prove that $\mathcal{Q}_N$ is a stepped quasi-plane of normal vector $\vec{\alpha}$.
But it is not clear that $D'_N$ is a suitable defect, \emph{i.e.}, satisfies the characterization of Prop.\ \ref{prop:suitable_defect}.
Therefore, we slightly modify $D'_N$ into the set $\tilde{D}_N'$ defined by
$$
\tilde{D}_N'=\{\vec{x}-k\vec{e}_i~|~\vec{x}\in D'_N,~k\in\mathbb{N},~\langle\vec{\alpha}|\vec{e}_i\rangle=0\},
$$
which satisfies the characterization of Prop.\ \ref{prop:suitable_defect}.
Moreover, $D'_N$ and $\tilde{D}_N'$ coincide on $B(\vec{0},R)$ (otherwise $\mathcal{Q}_N$ would not coincide on this ball with a stepped surface, namely $\mathcal{S}$).
Thus, by replacing $D'_N$ by $\tilde{D}_N'$ in the expression of $\mathcal{Q}_N$, we get a new stepped function $\tilde{\mathcal{Q}}_N$, which is a stepped quasi-plane of normal vector $\vec{\alpha}$ and satisfies:
$$
d_{\mathfrak{F}}(\mathcal{Q}_N,\tilde{\mathcal{Q}}_N)\leq 2^{-R}=\varepsilon/2.
$$
Finally, $\mathcal{S}$ is at distance at most $\varepsilon$ from $\tilde{\mathcal{Q}}_N$.
The result follows.
\end{proof}

\section{Additional remarks} \label{sec:con}
Let us summarize the strategy we have followed for defining a continued fraction expansion of stepped planes and stepped surfaces.
We start from a multi-dimensional continued fraction algorithm formulated with unimodular matrices, namely the Brun algorithm.
We then interpret these unimodular matrices as incidence matrices of well-chosen morphisms, namely Brun substitutions, with this choice being highly non-canonical.
We finally use the formalism of dual maps to associate with these unimodular matrices geometric maps acting on stepped planes and stepped surfaces.\\

Let us stress that we are not only able to substitute, \emph{i.e.}, to replace facets of hypercubes by unions of facets, but also to desubstitute, \emph{i.e.}, to perform the converse operation, by using the algebraic property $E_1^*(\sigma)^{-1}= E_1^*(\sigma^{-1})$.
We thus define a desubstitution process on geometric objects: local configurations determine the choice of the Brun substitution whose inverse dual map is applied.
We then show that any infinitely desubstituable stepped surface is almost a stepped plane (Th.\ \ref{th:weak_cv_surfaces}).
In particular, there are thus very few such surfaces.
This can seem disappointing, but this suggests an effective way to check planarity of a given stepped surface by computing its Brun expansion, since the longer this expansion is, the more planar the surface is.
In \cite{dgci_long}, we rely on the theoretical background here provided to obtain original algorithms for both \emph{digital plane recognition} and \emph{digital plane generation} problems.\\

We also plan to extend our approach to higher codimensions.
Indeed, stepped planes or surfaces here considered can be seen as codimension $1$ canonical projection tilings.
Roughly speaking, a dimension $d$ and codimension $k$ canonical projection tiling is a tiling of $\mathbb{R}^{d-k}$ obtained by projecting onto $\mathbb{R}^{d-k}$ the $(d-k)$-dimensional unit facets lying in a ''slice'' $V+[0,1)^d$ of $\mathbb{R}^d$, where $V\subset\mathbb{R}^d$ is a $(d-k)$-dimensional affine space (see, for more details, \cite{moody,sene}).
A first step in this direction is provided in \cite{AFHI}.\\

Last, note that if a stepped plane has a purely periodic Brun expansion, then it is a fixed-point of a dual map, namely the dual map of the composition of the Brun substitutions associated with this Brun expansion.
It is however unclear whether this property characterizes fixed-point stepped planes, mainly because of the lack of a Lagrange-like theorem for multi-dimensional continued fraction algorithms (see the discussion in \cite{ijfcs}).
A way to tackle this problem could be to extend in our multi-dimensional framework the approach of \cite{durand}.
We also would like to characterize linearly recurrent stepped planes in terms of their Brun expansions.

\paragraph{Acknowledgements.}
We would like to warmly thank the anonymous referees  of a  previous version of the present paper for their careful reading and their numerous and very valuable remarks (in particular, the clearer proof of Prop.\ \ref{prop:stepped_surface}, here given).


\begin{thebibliography}{blu}

\bibitem[AI01]{AI} P.~Arnoux, S.~Ito, \emph{Pisot substitutions and Rauzy fractals}, Bull. Bel. Math. Soc. Simon Stevin {\bf 8} (2001), 181--207.

\bibitem[AN93]{an} P.~Arnoux, A.~Nogueira, \emph{Mesures de Gauss pour des algorithmes de fractions continues multidimensionnelles}, Ann. Sci. \'Ecole Norm. Sup. {\bf 26} (1993), 645--664.


\bibitem[AIS01]{AIS} P.~Arnoux, S.~Ito, Y.~Sano, \emph{Higher dimensional extensions of substitutions and their dual maps}, J. Annal. Math. {\bf 83} (2001), 183--206.

\bibitem[ABFJ07]{ABFJ} P.~Arnoux, V.~Berth\'e, Th.~Fernique, D.~Jamet, \emph{Generalized substitutions, functional stepped surfaces and flips}, Theoret. Comput. Sci. {\bf 380} (2007), 251--267.

\bibitem[AFHI11]{AFHI} P.~Arnoux, M. Furukado, E. Harriss, S. Ito, \emph{Algebraic numbers, free group  automorphisms  and substitutions of the plane},  Transactions of the Amer. Math. Soc. {\bf 363} (2011), 4651--4699.

\bibitem[BM00]{moody} M.~Baake and R.~V.~Moody (eds.), \emph{Directions in Mathematical Quasicrystals}, CRM Monograph Series {\bf 13} (2000), 61--93, Amer. Math. Soc.: Providence, RI.

\bibitem[BK06]{BK} M.~Barge, J.~Kwapisz, \emph{Geometric theory of unimodular Pisot substitutions}, Amer. J. Math. {\bf 128} (2006), 1219--1282.

\bibitem[BV00]{BV} V.~Berth\'e, L.~Vuillon, \emph{Tilings and rotations on the torus: a two-dimensional generalization of {S}turmian sequences}, Disc. Math. {\bf 223} (2000), 27--53.

\bibitem[BFR08]{BFR} O.~Bodini, Th.~Fernique, \'E.~R\'emila, \emph{Flip-accessibility of rhombus tilings of the whole plane}, Inf. Comput. {\bf 206} (2008), 1065--1073.

\bibitem[Bre81]{brentjes} A.~J.~Brentjes, \emph{Multi-dimensional continued fraction algorithms}, Mathematical Centre Tracts \textbf{145}, Matematisch Centrum, Amsterdam, 1981.

\bibitem[Bru57]{brun} V.~Brun, \emph{Algorithmes euclidiens pour trois et quatre nombres}, 13th Congr. Math. Scand. Helsinki (1957), 45--64.
 
\bibitem[Dur03]{durand} F.~Durand, \emph{Linearly recurrent subshifts have a finite number of non-periodic subshift factors}, Ergod. Th. $\&$ Dynam. Sys. {\bf 20} (2000), 1061--1078. \emph{Corrigendum and addendum to: Linearly recurrent subshifts have a finite number of non-periodic factors}, Ergod. Th. $\&$ Dynam. Sys. {\bf 23} (2003), 663--669.

\bibitem[Ei03]{ei} H.~Ei, \emph{Some properties of invertible substitutions of rank d and higher dimensional substitutions}, Osaka J. Math. {\bf 40} (2003), 543--562.

\bibitem[Fer06]{ijfcs} Th.~Fernique, \emph{Multi-dimensional Sequences and Generalized Substitutions}, Int. J. Fond. Comput. Sci. \textbf{17} (2006), 575--600.

\bibitem[Fer09]{dgci_long} Th.~Fernique, \emph{Generation and recognition of digital planes using multi-dimensional continued fractions}, Pattern Recognition {\bf 42} (2009), 2229--2238.

\bibitem[IO93]{IO93} S.~Ito, M.~Ohtsuki, \emph{Modified {J}acobi-{P}erron algorithm and generating {M}arkov partitions for special hyperbolic toral automorphisms}, Tokyo J. Math. {\bf 16}, (1993), 441--472.

\bibitem[IFHY03]{ifhy} S.~Ito, J.~Fujii, H.~Higashino, S.-I.~Yasutomi, \emph{On simultaneous approximation to {$(\alpha,\alpha\sp 2)$} with {$\alpha\sp 3+k\alpha-1=0$}}, J. Number Theory {\bf 99} (2003), 255--283.

\bibitem[Jam04]{jamet} D.~Jamet, \emph{On the language of standard discrete planes and surfaces}, in Proc. of IWCIA'04 (2004), 232--247.

\bibitem[Lag94]{lagarias} J.~C.~Lagarias, \emph{Geodesic multidimensional continued fractions}, Proc. London Math. Soc. {\bf 69} (1994), 464--488.

\bibitem[Lot02]{Loth} N.~Lothaire, \emph{Algebraic combinatorics on words}, Cambridge University Press, 2002.
 
\bibitem[Pyt02]{Pyt} N.~Pytheas Fogg, \emph{Substitutions in Dynamics, Arithmetics, and Combinatorics}, Lecture Notes in Mathematics {\bf 1794}, Springer Verlag. V. Berth\'e, S. Ferenczi, C. Mauduit and A. Siegel, Eds. (2002).

\bibitem[RW92]{RW} C.~Radin, M.~Wolff, \emph{Space tilings and local isomorphism}, Geometriae Dedicata {\bf 42} (1992), 355--360.

\bibitem[Rev91]{reveilles} J.-P.~R\'eveilles, \emph{Calcul en nombres entiers et algorithmique}, Th\`ese d'\'Etat, Univ. Louis Pasteur, Strasbourg (1991).

\bibitem[Rob96]{rob} E.~A.~Robinson Jr., \emph{The dynamical theory of tilings and quasi\-cryst\-al\-lography}, Ergodic theory of ${\mathbb Z}^d d$ actions, London Math. Soc. Lecture Note Ser. {\bf 228} (1996), 451--473. Cambridge Univ. Press. 

\bibitem[Sch00]{schweiger} F.~Schweiger, \emph{Multi-dimensional continued fractions}, Oxford Science Publications, Oxford Univ. Press, Oxford (2000).

\bibitem[Sen95]{sene} M.~Senechal, \emph{Quasicrystals and geometry}, Cambridge University Press, Cambridge (1995).

\bibitem[Thu89]{thurston} W.~P.~Thurston, \emph{Groups, tilings and finite state automata}, Lectures notes distributed in conjunction with the {C}olloquium {S}eries, in {AMS} {C}olloquium lectures, 1989.

\end{thebibliography}
\end{document}